\documentclass[journal]{IEEEtran}
\usepackage[utf8]{inputenc}
\usepackage{amsmath}
\usepackage{soul}
\usepackage{array}
\usepackage{amsmath,amsthm,amssymb,epsf}
\usepackage{bm} %% Bold Math
\usepackage{graphicx,psfrag}
\usepackage{epstopdf}
\usepackage{verbatim}
\usepackage{multirow}
\usepackage{color}
\usepackage{setspace}
\usepackage{enumitem}
\usepackage{subfigure}
\usepackage{amssymb}
\usepackage{tabularx}
\usepackage{booktabs}
\usepackage{dirtytalk}
\usepackage{amsmath}
{
      \theoremstyle{plain}
      \newtheorem{assumption}{Assumption}
  }
\usepackage{csquotes}
\usepackage{algorithm}
\usepackage{algorithmic}

\newtheorem{theorem}{Theorem}
\newtheorem{definition}{Definition}
\newtheorem{problem}{Problem}
\newtheorem{lemma}{Lemma}

\newtheorem{remark}{Remark}
\newtheorem{condition}{Condition}
\newtheorem{corollary}{Corollary}
\usepackage{soul}
\usepackage{cite}

\begin{document}

\title{A Unified Framework for Online Data-Driven Predictive Control with Robust Safety Guarantees}

\author{Amin~Vahidi-Moghaddam, Kaian~Chen, Kaixiang~Zhang, Zhaojian~Li$^*$, Yan~Wang, and Kai~Wu
%,~\IEEEmembership{Senior Member,~IEEE,}
\thanks{This work was supported by the Ford Motor Company under award number 000315-MSU0173.}
\thanks{$^*$Zhaojian Li is the corresponding author.}
\thanks{Amin Vahidi-Moghaddam, Kaian Chen, Kaixiang Zhang, and Zhaojian Li are with the Department of Mechanical Engineering, Michigan State University, East Lansing, MI 48824 USA (e-mail: vahidimo@msu.edu, chenkaia@msu.edu, zhangk64@msu.edu, lizhaoj1@egr.msu.edu).}
\thanks{Yan Wang and Kai Wu are with the Research and Advanced Engineering, Ford Motor Company, Dearborn, MI 48121 USA (e-mail: ywang21@ford.com, kwu41@ford.com).}
% <-this % stops a space
} 

\maketitle

\begin{abstract}
Despite great successes, model predictive control (MPC) relies on an accurate dynamical model and requires high onboard computational power, impeding its wider adoption in engineering systems, especially for nonlinear real-time systems with limited computation power. These shortcomings of MPC motivate this work to make such a control framework more practically viable for real-world applications. Specifically, to remove the required accurate dynamical model and reduce the computational cost for nonlinear MPC (NMPC), this paper develops a unified online data-driven predictive control pipeline to efficiently control a system with guaranteed safety without incurring large computational complexity. The new aspect of this idea is learning not only the real system but also the control policy, which results in a reasonable computational cost for the data-driven predictive controllers.
More specifically, we first develop a spatial temporal filter (STF)-based concurrent learning scheme to systematically identify system dynamics for general nonlinear systems. We then develop a robust control barrier function (RCBF) for safety guarantees in the presence of model uncertainties and learn the RCBF-based NMPC policy. Furthermore, to mitigate the performance degradation due to the existing model uncertainties, we propose an online policy correction scheme through perturbation analysis and design of an ancillary feedback controller. Finally, extensive simulations on two applications, cart-inverted pendulum and automotive powertrain control, are performed to demonstrate the efficacy of the proposed framework, which shows comparable performance with much lower computational cost in comparison with several benchmark algorithms.
\end{abstract}


\IEEEpeerreviewmaketitle

\section{Introduction}
\label{Sec1}
Attaining peak performance while guaranteeing safety is a common goal in many engineering systems. Model predictive control (MPC) arises as a promising framework to accomplish this task by solving a constrained optimization problem with future state predictions \cite{zidek2021model,alonso2021robust,lopez2019dynamic,wytock2017dynamic}. However, MPC relies on an accurate dynamical model, which is challenging to obtain, especially for nonlinear and complex systems \cite{vahidi2023extended,vahidi2022event,yazdandoost2022optimization,hajidavalloo2021mpc}. Towards that end, in our prior work \cite{chen2020online}, we developed an online nonlinear system identification that exploits spatial-temporal filters (STF) to systematically decompose a nonlinear system into local composite models using input-output data. Compared to black-box models, using neural network (NN) \cite{akbari2023blending,ahmadi2022deep,akbari2022numerical,jamali2021objective,jebellat2021training} and Gaussian process regression (GPR) \cite{yousefpour2023unsupervised,foumani2023multi}, the STF has simpler form with greater interpretability, and the resultant MPC has a reasonable performance \cite{chen2022stochastic}. Despite promising empirical results, the STF needs to satisfy the persistence of excitation (PE) conditions, which is very important in learning and will be treated in this paper.  

In addition, the typical MPC implementation involves solving an online optimization problem at each step and thus requires high onboard computational capabilities that are not generally available in many embedded systems. To tackle this issue, one prevalent approach involves the utilization of model-reduction techniques to simplify the system dynamics \cite{zhang2022dimension,amiri2022fly,zhong2022optimally,lore2021model}. However, employing such techniques necessitates a trade-off between system performance and computational complexity, and even after model reduction, the computational burden often remains substantial. Another sound approach involves the implementation of function approximators, such as the NNs \cite{bao2022learning,krishnamoorthy2021adaptive} and the GPR \cite{arcari2023bayesian,arcari2020meta}, to effectively learn the MPC policy. However, the trained controllers suffer from several shortcomings: i) Unlike the MPC that ensures system safety, the trained controllers offer no assurance on safety, and ii) The trained controllers inevitably cause performance degradation due to control learning error, system identification error, and unknown disturbance. These issues have significantly hindered the adoption of such controllers in real-world engineering applications.
 
As such, several lines of work has been conducted to address the aforementioned challenges. For example, to guarantee system safety, control barrier function (CBF) has recently emerged as a promising framework to efficiently handle constraints \cite{dean2020guaranteeing,chen2020safety,ames2019control}. The CBF implies forward invariance of a safety set based on system dynamics \cite{isaly2020zeroing,shivam2020intersection}. Moreover, robust CBF (RCBF) and adaptive CBF (ACBF) approaches have been proposed to maintain system safety in the presence of system uncertainties such as unknown disturbance, model mismatch, and state estimation error \cite{garg2021robust,black2021fixed,isaly2021adaptive,lopez2020robust,taylor2020adaptive}. On the other hand, to minimize the performance loss caused by the controller learning error, an online adaptive control policy has been proposed for the MPC policy learning \cite{krishnamoorthy2021adaptive}. However, a holistic treatment of the control learning error, the system identification error, and the unknown disturbance -- necessary for real-world engineering systems -- has not been developed.

In this paper, we develop a unified data-driven safe predictive control framework that efficiently regulates system states with guaranteed safety for general nonlinear systems. Specifically, we first develop a discrete-time STF-based concurrent learning for efficient online nonlinear system identification. Compared to noise-injection PE satisfaction schemes, the concurrent learning uses a memory of collected data to propose a rank condition instead of the PE condition, which makes it easy to monitor online and removes the requirement of noise signal \cite{vahidi2021learning,vahidi2020memory,chowdhary2010concurrent}. An extended RCBF scheme is further developed to guarantee system safety by systematically considering all model uncertainties. A nonlinear MPC (NMPC) is then developed based on the composite local model structure and the RCBF constraint, and another STF function approximator is exploited to learn the NMPC policy. Finally, a policy correction scheme is proposed for efficient online implementation.

The novelty and contributions of this paper include the following. First, the proposed discrete-time STF-based concurrent learning technique handles both structured and unstructured uncertainties as well as unknown external disturbance without requiring derivatives of system states or filter regressors to remove the PE condition compared to \cite{chowdhary2010concurrent,vahidi2020memory,vahidi2021learning}. Second, the developed RCBF guarantees system safety in the presence of not only the system identification error and the external disturbance but also the control learning error. Different from \cite{garg2021robust,black2021fixed,isaly2021adaptive,lopez2020robust,taylor2020adaptive}, this framework provides a holistic treatment to guarantee system safety in the presence of control learning error, system identification error, and unknown disturbance -- necessary for the real-world engineering systems. Third, an online adaptive control policy, including a KKT adaptation and a feedback control, is proposed to mitigate the performance loss due to the existing approximation errors such that it keeps the real trajectory around the ideal nominal trajectory. Last but not least, the efficacy of the proposed control synthesis is demonstrated in applications of cart-inverted pendulum and automotive powertrain control.

The remainder of this paper is organized as follows. The problem formulation and  preliminaries on the MPC and the CBF are provided in Section~II. Section~III presents the proposed online STF-based data-driven safe predictive control. Section~IV presents simulation results on the cart inverted pendulum and the automotive powertrain control. Finally, conclusions are drawn in Section~V.

\textbf{Notations}. Throughout the paper, the following notations are adopted. $\mathbb{R}^n$ and $\mathbb{R}^{n\times m}$ denote the set of $n$-dimensional real vectors and the set of $n\times m$-dimensional real matrices, respectively. $A>0$ denotes that the matrix $A$ is positive definite. $x^T$ and $A^T$ denote the transpose of the vector $x$ and the matrix $A$, respectively. $\lVert . \rVert$ represents the Euclidean norm of a vector or the induced $2$-norm of a matrix. $\lambda_{min} (A)$ and $\lambda_{max} (A)$ denote the minimum eigenvalue and the maximum eigenvalue of $A$, respectively. The identity matrix with $n$-dimension is denoted by $I_n$. Moreover, Table I represents the used abbreviations in this paper.

\begin{table}[!ht]
\centering
 \caption{Abbreviations}
\begin{tabular}{ |p{2cm}|p{6cm}| }
\hline\hline
Abbreviations & Meaning \\
\hline
MPC & Model Predictive Control \\\hline
NMPC & Nonlinear Model Predictive Control \\\hline
STF & Spatial-Temporal Filter \\\hline
NN & Neural Network \\\hline
GPR & Gaussian Process Regression \\\hline
PE & Persistence of Excitation \\\hline
CBF & Control Barrier Function \\\hline
RCBF & Robust Control Barrier Function \\\hline
ACBF & Adaptive Control Barrier Function \\\hline
CL & Concurrent Learning \\\hline
RLS & Recursive Least Squares \\\hline
UUB & Uniformly Ultimately Bounded \\\hline
QP & Quadratic Programming \\\hline
\hline
\end{tabular}
\end{table}

\section{Preliminaries and Problem Formulation}
\label{Sec2}
In this section, preliminaries on model predictive control (MPC) and control barrier function (CBF) are first reviewed, and problem of online data-driven predictive control with safety guarantees is then formulated for nonlinear discrete-time systems.

\subsection{Model Predictive Control}
Consider a class of nonlinear discrete-time system that has the following form: 
\begin{equation}
    \label{system}
    x(k+1) = f(x(k),u(k))+w(k),\quad y(k) = g(x(k)),\\
\end{equation}
where $k \in \mathbb{N}^+$ denotes the time step, $x \in \mathbb{R}^n$ is the state vector, $u \in {\mathbb{R}^m}$ represents the control input, $w \in \mathbb{R}^n$ is an unknown external disturbance, and $y \in \mathbb{R}^l$ denotes the output of the system. Moreover, $f:\mathbb{R}^n\times \mathbb{R}^m \rightarrow \mathbb{R}^n$ is the system dynamics with $f(0,0)=0$, and $g:\mathbb{R}^n \rightarrow \mathbb{R}^l$ represents the output dynamics. We assume  the system states are measurable.

The system state and controls are subject to the following constraints:
\begin{equation}
    \label{Constraints}
    u(k) \in \mathbb{U} \subset \mathbb{R}^m,
     \quad x(k) \in \mathbb{X} \subset \mathbb{R}^n,
\end{equation}
where the input constraints are generic subsets of $R^m$, and the state constraints can be represented as a set of nonlinear inequalities, including box, ellipsoidal, or polytopic constraints.

%Now, the following definition describes the closed-loop performance for the nonlinear system \eqref{system}.

\begin{definition}[Closed-Loop Performance]
\label{def1}
Consider the nonlinear system \eqref{system} and a control problem of tracking a desired time-varying reference $r$ by the output $y$. Starting from an initial state $x_0$, the closed-loop system performance over $N$ steps is characterized by the following cost term:
\begin{equation}
    \label{Cost}
    \begin{aligned}
  &J_N(\bf{x},\bf{u}) \\&= \sum^{N-1}_{k=0} \phi(x(k),u(k),y(k),r(k)) + \psi(x(N),y(N),r(N)),
  \end{aligned}
\end{equation}
where $\mathbf{u} = \left[ u(0),\, u(1),\, \cdots,\, u(N-1)  \right]$, $\mathbf{x} = \left[ x(0),\, x(1),\, \cdots,\, x(N)  \right]$, and $\phi(x,u,y,r)$ and $\psi(x,y,r)$ are respectively the stage cost and the terminal cost that take the following forms:
\begin{equation}
  \begin{aligned}
    \label{Cost format}
    & \phi(x,u,y,r) = x^{T} Q x + u^{T} R u + (y-r)^{T} P (y-r), \\
    & \psi(x,y,r) = x^{T} Q_N x + (y-r)^{T} P_N (y-r),
  \end{aligned}
\end{equation}
where $Q$, $R$, $P$, $Q_N$, and $P_N$ are positive-definite matrices of appropriate dimensions.
\end{definition}

With the defined cost, the control goal is now to minimize   the cost function \eqref{Cost} while adhering to the constraints in \eqref{system} and \eqref{Constraints}. %It is worth noting that we have best closed-loop performance for the nonlinear system \eqref{system} if we minimize the cost function \eqref{Cost}. Therefore, the MPC aims at minimizing the cost function \eqref{Cost} while adhering to the constraints.
In practice, the real nonlinear system \eqref{system} may not be available; thus, system identification algorithms are typically used to identify the system model. We denote the identified (nominal) model as
\begin{equation}
    \label{nominal model}
     \hat{x}(k+1) = \hat{f}(\hat{x}(k), u(k)),\quad
     \hat{y}(k) = g(\hat{x}(k)),
\end{equation}
where $\hat{x}$, $\hat{y}$, and $\hat{f}$ denote the  states,  outputs, and dynamics of the identified model, respectively.

\begin{comment}
\begin{remark} [Nominal Control]
\label{Nominal Control}
One can apply the designed nominal control $u_{nm}$ to the real system \eqref{system}, i.e. $u=u_{nm}$; however, it may cause a poor performance for the real system due to the system identification error and the unknown disturbance. Thus, $u_{nm}$ should be modified to use in the real system. 
\end{remark}
\end{comment}

The MPC aims at optimizing the system performance over $N$ future steps using the nominal model \eqref{nominal model}, and at each time step $k$, it is reduced to the following constrained optimization problem:
\begin{equation}
  \begin{aligned}
    \label{NMPC}
    &\qquad\qquad(\mathbf{x}^{*},\mathbf{u}^{*}) = \underset{\mathbf{\hat{x}},\mathbf{u}}{\arg\min} \hspace{1 mm} J_N(\mathbf{\hat{x}},\mathbf{u})\\
    &s.t.\quad \hat{x}(k+1) = \hat{f}(\hat{x}(k),u(k)),\quad \hat{y}(k) = g(\hat{x}(k))\\
    & \hspace{8.5 mm} \hat{x}(0) = x(k),\quad u(k) \in U, \quad \hat{x}(k) \in X,
  \end{aligned}
\end{equation}
%where    $\mathbf{u}^{*} = \left[ u^{*}_0, u^{*}_1, ..., u^{*}_{N-1}  \right]$ and $\;    \mathbf{x}^{*} = \left[ x^{*}_0, x^{*}_1, ..., x^{*}_N  \right]$.
In a typical MPC implementation, only the first optimal control $u^*(0)$ is executed, the system evolves one step, and the process is then repeated.
The MPC has enjoyed great success in numerous applications; however, there are still a few challenges that remain, including high computational cost (especially for nonlinear systems) and robustness to the model uncertainties, which will be treated in Section~III.

\subsection{Control Barrier Function}
%In this subsection, the preliminaries of the discrete-time CBF is presented, which can be used for safety-critical controls.
Consider a closed set $S$ that is defined as the sublevel set of a continuously differentiable function $h: \mathbb{X} \subset \mathbb{R}^n \rightarrow \mathbb{R}$:
\begin{equation}
\begin{aligned}
    \label{Safe set}
  & S = \{x \in \mathbb{R}^n: h(x) \leq 0\},\\
  & \partial S = \{x \in \mathbb{R}^n: h(x) = 0\},\\
  & Int(S) = \{x \in \mathbb{R}^n: h(x) < 0\}.
  \end{aligned}
\end{equation}
The set $S$ can be used to define the safe set of a system and is forward invariant if the function $h(x)$ is a CBF, i.e.,
\begin{equation}
\begin{aligned}
    \label{CBF condition}
  \triangle h(x(k)) \leq -\gamma h(x(k)),\quad 0 < \gamma \leq 1,
  \end{aligned}
\end{equation}
where $\triangle h(x(k)):=h(x(k+1))-h(x(k))$, and $\gamma$ is a design parameter. Using the CBF constraint \eqref{CBF condition}, it can be shown that the upper bound of the CBF $h(x)$ decreases exponentially with a rate of $1-\gamma$ as
\begin{equation}
\begin{aligned}
    \label{CBF condition2}
  h(x(k+1)) \leq (1-\gamma) h(x(k)).
  \end{aligned}
\end{equation}
The CBF is widely used in engineering systems to avoid unsafe regions and guarantee system safety; however, robustness to the model uncertainties must be addressed, which will be treated in Section III-B.

\subsection{Problem Statement}
The MPC scheme in \eqref{NMPC} poses two major challenges. First, the nonlinear constrained optimization problem is solved at each time step, incurring heavy computations that are often too computationally expensive for systems with limited onboard computations and real-time constraints. Second, the control performance highly relies on the accuracy of the identified model \eqref{nominal model}. As such, the goal of this paper is to develop a unified data-driven predictive control framework to address the challenges. Specifically, to overcome the first issue, we aim to learn the nonlinear MPC (NMPC) policy using a function approximator such that we have $\tilde{u}\approx\pi_{\text{MPC}}$, where $\tilde{u}$ and $\pi_{\text{MPC}}$ represent the NMPC policy approximation and the NMPC policy, respectively. To address the second issue, we aim at developing an online adaptation scheme to robustly handle the model uncertainties. More specifically, from \eqref{system}, \eqref{nominal model}, and $\tilde{u}$, the following drift dynamics errors are defined:
\begin{equation}
  \begin{aligned}
    \label{identification error}
    &e_{s}(k) = \hat{f}(\hat{x}(k),u(k)) - f(x(k),u(k)),\\ &e_{c}(k) = \hat{f}(\hat{x}(k),\tilde{u}(k)) - \hat{f}(\hat{x}(k),u(k)),
    \end{aligned}
\end{equation}
where $e_{s}$ and $e_{c}$ represent the drift dynamics errors due to the imperfections on the system identification and the NMPC policy learning, respectively. As a result, one can express the real system \eqref{system} in the following form:
\begin{equation}
  \begin{aligned}
    \label{disturbed system}
    x(k+1) = \hat{f}(\hat{x}(k),\tilde{u}(k)) - e_{c}(k) - e_{s}(k) + w(k).
  \end{aligned}
\end{equation}

Note that due to the system identification error $e_{s}$ and the unknown disturbance $w$, the NMPC \eqref{NMPC} may not optimize the closed-loop performance \eqref{Cost} for the real system \eqref{system}. Moreover, the control policy learning error $e_{c}$ further exacerbates the complexity to achieve a satisfactory performance for the real system. Therefore, the function-approximated control policy $\tilde{u}$ needs to  adapt online so that the performance loss, caused by $e_{c}$, $e_{s}$, and $w$, is mitigated. The objective is thus to develop an online data-driven safe predictive control such that it efficiently identifies the real system and subsequently learns a discrete-time RCBF-based NMPC policy offline, and then adapts the NMPC policy approximation online to optimize the closed-loop performance for the real system with the existing model uncertainties. This objective is stated as follows.

\begin{problem} [Online Data-Driven Safe Predictive Control]
\label{prob1}
Consider the nonlinear system \eqref{disturbed system} with the safety constraint \eqref{CBF condition2}. Design an online data-driven safe predictive control to achieve the following properties:

\hspace{-4.75 mm} i) For the offline part, $e_{s}$ and $e_{c}$ converge to zero if $w=0$ and to a bounded region around zero if $w \neq 0$.

\hspace{-5 mm} ii) For the online part, the safety constraint \eqref{CBF condition2} is guaranteed for the real system.

\hspace{-5.25 mm} iii) For the online part, the appeared performance loss due to $e_{c}$, $e_{s}$, and $w$ is mitigated.
\end{problem}

%To solve Problem \ref{prob1}, the following assumption and definition will be used in the proposed approach:

\begin{assumption} [Bounded Terms]
\label{Bounded Terms}
1) $f(x,u)$ is bounded for bounded inputs, 2) There exists $\varepsilon_w> 0$ such that $\|w(k)\| \leq \varepsilon_w$ for all time steps $k$, and 3) $\eta$ is the Lipschitz constant of the discrete CBF $h(x)$ such that
\begin{equation}
  \begin{aligned}
    \label{DT Lipschitz constant}
  \left| h(x)-h(\hat{x}) \right| \leq \eta \left\| x-\hat{x} \right\|.  
  \end{aligned}
\end{equation}
\end{assumption}

\begin{definition} [Persistently Exciting \cite{bai1985persistency}]
\label{def pe}
The bounded vector signal $Z(\cdot) \in \mathbb{R}^n$ is persistently exciting (PE), in $\bar{k}$ steps if there exists $\bar{k}\in \mathbb{Z}^+$ and $\beta>0$ such that 
\begin{equation}
  \begin{aligned}
    \label{PE}
    &\sum^{\sigma + \bar{k}}_{k=\sigma} Z(k) Z^{T}(k) \geq \beta I,
  \end{aligned}
\end{equation}
where $\sigma \geq k_0$ with $k_0 \in \mathbb{N}$ being the initial time step.
\end{definition}

\begin{figure}[!ht]
     \centering
     \includegraphics[width=0.82\linewidth]{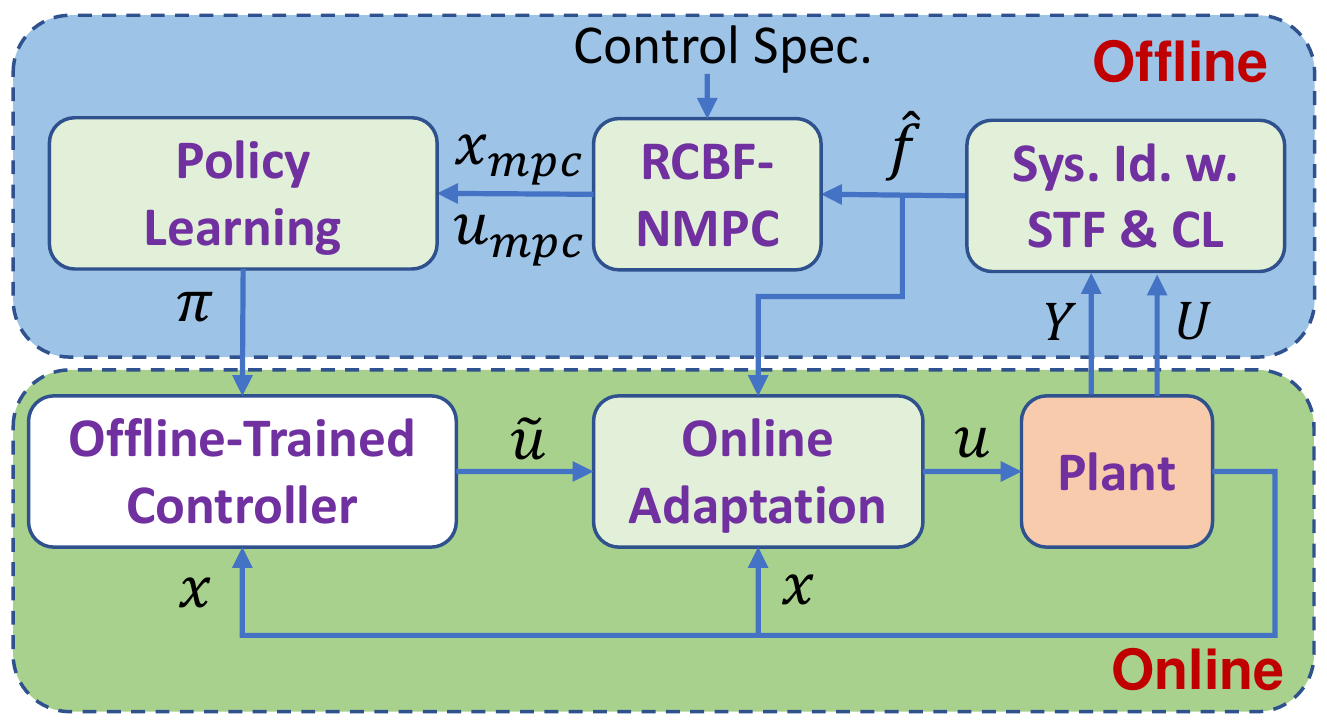}
     \caption{Schematic of online data-driven safe predictive control.}
     \label{block diagram}
\end{figure}

\section{Main Results}
\label{Sec3}
The developed data-driven safe predictive control pipeline is shown in Fig.~\ref{block diagram}. First, the discrete-time concurrent learning (CL) technique is integrated with the spatial temporal filter (STF)  \cite{chen2020online} for efficient nonlinear system identification with relaxed PE conditions. With the identified model and control specifications, the RCBF-based NMPC is exploited to optimize the system performance such that the safety is guaranteed despite the control policy learning error, the system identification error, and the unknown disturbances. A function approximator is then trained to approximate the RCBF-based NMPC policy for efficient online implementation, and an adaptation scheme is developed to minimize the performance loss due to the learning errors.
The main components of the framework are detailed in the following subsections.

\subsection{Nonlinear System Identification with STF-based Concurrent Learning}
In this subsection, we present a unified nonlinear system identification to obtain \eqref{nominal model} and enable the nominal NMPC design. Specifically, we identify a nonlinear autoregressive exogenous model (NARX) using the input-output data as follows
\begin{equation}
  \begin{aligned}
    \label{NARX Model}
    y(k)&=G(U_{d}(k),Y_{d}(k)), \\ 
    U_{d}(k)&=[u^{T}(k-1),u^{T}(k-2),\ldots ,u^{T}(k-d_{u})], \\ 
    Y_{d}(k)&=[y^{T}(k-1),y^{T}(k-2),\ldots ,y^{T}(k-d_{y})].
  \end{aligned}
\end{equation}
Here $d_{u}$ and $d_{y}$ are the input delay and the output delay, respectively, and $G$ is a nonlinear output prediction function. %Moreover, $n_u$ and $n_y$ denote the dimensions of the control input $u$ and the output $y$, respectively.

Towards that end, we follow our prior work on STF-based system identification that exploits evolving clustering and recursive least squares to systematically decompose the nonlinear system into multiple local models and simultaneously identify the validity zone of each local model   \cite{chen2020online}. Specifically, let $U_{stf}(k) = [U_{d}(k),Y_{d}(k)]^{T}$%, $n_U=n_u \times d_{u} + n_y \times d_{y}$ represents the dimension of $U_{stf}$. Now, 
, the nonlinear model \eqref{NARX Model} is expressed as a composite local model structure, where each local model has a certain valid operating regime, in the following form:
\begin{equation}
  \begin{aligned}
    \label{Composite Structure}
  y(k)&=F(U_{stf}(k),\omega(k);\Phi ,\Psi) \\ 
  &=\sum\limits_{i=1}^{L} {\alpha }_{i}({U}_{stf}(k);\phi_i,{\psi }_{i}) {f}_{i}({U}_{stf}(k),\omega_i(k);{\phi }_{i}), 
  \end{aligned}
\end{equation}
where ${f}_{i}({U}_{stf}(k),\omega_i(k);{\phi }_{i})$ is the $i$th local model that can take the form of a point, linear model, Markov chain, neural network etc, parameterized by ${\phi }_{i}$ in the presence of unstructured uncertainty $\omega_i(k)$. %For example, ${\phi }_{i}$ represents the probability transition matrix if ${f}_{i}({U}_{stf}(k),{\phi }_{i},\omega_i(k))$ is considered as a Markov chain, and ${\phi }_{i}$ represents the regression coefficients and biases if ${f}_{i}({U}_{stf}(k),{\phi }_{i},\omega_i(k))$ is a linear model.
${\alpha }_{i}({U}_{stf}(k);{\psi }_{i})$ is the weighting functions that interpolates the local models and is parameterized by ${\psi }_{i}$. It is based on a dissimilarity metric that combines clustering \cite{boroujeni2023hybrid,boroujeni2021data,derakhshan2021detecting} and local model prediction error measures. $L$ is the number of local models, and  $\Phi=[{\phi }_{1},{\phi }_{2},\ldots,{\phi }_{L}]$, $\Psi=[{\psi }_{1},{\psi }_{2},\ldots,{\psi }_{L}]$, and $\omega(k)=[w_{1}(k),w_{2}(k),\ldots,w_{L}(k)]$ are the collection of local model parameters, local interpolating function parameters, and unknown local unmodeled disturbance $w(k)$, respectively. 

In particular, we consider linear local models and a softmax-like interpolation function: 
\begin{equation}
  \begin{aligned}
    \label{local models}
  &{f}_{i}({U}_{stf}(k),\omega_i(k);{\phi }_{i}) = {A}_{i}{U}_{stf}(k)+{b}_{i}+\omega_i(k), \\ 
  &{\alpha }_{i}({U}_{stf}(k);{\psi }_{i},\phi_i) = %\frac{{{\mu }_{i}}({U}_{stf}(k),{\psi }_{i})}{\sum\limits_{j=1}^{K}{{\mu }_{j}({U}_{stf}(k),{\psi }_{j})}}. 
  \frac{\exp(-D_i({U}_{stf});{\psi }_{i},\phi_i ))}{\sum \limits_{j=1}^{L}{\exp (-D_j({U}_{stf});{\psi }_{i},\phi_i))}}.
  \end{aligned}
\end{equation}
Here  $A_i \in \mathbb{R}^{n_y \times n_U}$ and $B_i \in \mathbb{R}^{n_y}$ are local model parameters, and $\omega_i(k) \in \mathbb{R}^{n_y}$ is the  unstructured uncertainty for local model $i$. Moreover, ${\alpha }_{i}({U}_{stf}(k);{\psi }_{i},\phi_i)$ resembles the softmax-like function, and $D_i({U}_{stf});{\psi }_{i},\phi_i )$ is the dissimilarity metric that combines a Mahalanobis distance and the model residual \cite{chen2020online}.
Essentially, the STF uses evolving clusters of ellipsoidal shape as function bases for local model interpolation. Each local model ${f}_{i}$ is associated with an evolving cluster such that the clusters and the local model parameters are updated simultaneously. The readers are referred to \cite{chen2020online} for more details about the STF.

The standard STF \cite{chen2020online} uses the recursive least squares (RLS) to update the local model parameters $\{A_i,b_i\}$, which requires a noise signal added to the system control input so that the PE condition is satisfied for learning error convergence\footnote{A video of the simulation result on the STF-Idnetifier can be found online at https://www.youtube.com/watch?v=UYZiNC1LJwM\&t=12s.}. However, this will inevitably increase the complexity and may not be applicable in systems where the control inputs are not completely programmable. Therefore, in this subsection, we propose a discrete-time concurrent learning technique for the STF to relax the PE condition with a rank condition that is convenient for inspection and implementations.  

Specifically, using the local linear models \eqref{local models}, the nonlinear model \eqref{Composite Structure} is rewritten as
\begin{equation}
  \begin{aligned}
    \label{Composite Structure b}
  &y(k)=\sum\limits_{i=1}^{L} {\alpha }_{i}({U}_{stf}(k);{\psi }_{i}) ({A}_{i}{U}_{stf}(k)+{b}_{i}+\omega_i(k)),
  \end{aligned}
\end{equation}
which can be further written in the following compact form:
\begin{equation}
  \begin{aligned}
    \label{Regressor Format}
  &y(k)=\Phi \zeta({U}_{stf}(k),\alpha (k))+\omega(k), 
  \end{aligned}
\end{equation}
where $ \Phi = \left[A_1, b_1,\ldots, A_L, b_L  \right]\in \mathbb{R}^{n_y \times q}$, $\zeta({U}_{stf},\alpha) = \left[\alpha_1 {U}^{T}_{stf},\, \alpha_1, \ldots, \alpha_L {U}^{T}_{stf},\, \alpha_L \right]^T\in \mathbb{R}^{q}$, and $q=L(n_U+1)$. Moreover, the total unknown disturbance is given as $\omega=\alpha_1 \omega_1+ \ldots+ \alpha_L \omega_L \in \mathbb{R}^{n_y}$, where $\|\omega(k)\| \leq \varepsilon$ and $\varepsilon=\alpha_1 \varepsilon_1+ \ldots+ \alpha_L \varepsilon_L$ with $\varepsilon_i$ being the bound of $\omega_i$, i.e.,  $\|\omega_i(k)\| \leq \varepsilon_i$.

Therefore, the identified model can be represented as
\begin{equation}
  \begin{aligned}
    \label{Regressor Identified Model}
  &\hat{y}(k)=\hat{\Phi}(k) \zeta({U}_{stf}(k),\alpha (k)), 
  \end{aligned}
\end{equation}
where $\hat{\Phi} = \left[\hat{A}_1, \hat{b}_1, \ldots, \hat{A}_L, \hat{b}_L  \right]$. Thus, at each time step $k$, the system identification error is defined as
\begin{equation}
  \begin{aligned}
    \label{System Identification Error}
  &{e}_{si}(k) = \hat{y}(k) - y(k)\\
  & \hspace{8.75 mm}= \tilde{\Phi}(k) \zeta({U}_{stf}(k),\alpha (k))-\omega(k),
  \end{aligned}
\end{equation}
where $\tilde{\Phi}(k) = \hat{\Phi}(k) - \Phi$ describes the parameter identification errors. Note that the system identification error ${e}_{si}$ is measurable since the output $y$ is available for measurement.

The concurrent learning technique records past data and collects them in a history stack as
\begin{equation}
  \begin{aligned}
    \label{Recorded Data}
  &Z = \left[\zeta({U}_{stf}(k_1),\alpha (k_1)), \ldots, \zeta({U}_{stf}(k_s),\alpha (k_s)) \right],
  \end{aligned}
\end{equation}
where $k_1, \ldots, k_s$ are the past historic recording time steps, and $s$ denotes the number of recorded data in the history stack. At the current time step $k$, the system identification error for the $j$th recorded sample $\zeta({U}_{stf}(k_j),\alpha (k_j))$ is given as
\begin{equation}
  \begin{aligned}
    \label{System Identification Error j}
  &{e}_{si}(k_j) = \tilde{\Phi}(k) \zeta({U}_{stf}(k_j),\alpha (k_j))-\omega(k_j), \hspace{1 mm} j = 1, 2, \ldots, s,  
  \end{aligned}
\end{equation}
where $\tilde{\Phi}(k)$ is the parameter identification error at the current time step. Then, defining the normalizing signal $m(k)=\sqrt[]{\varrho+\zeta(k)^T \zeta(k)}, \varrho>0$, one has
\begin{equation}
  \begin{aligned}
    \label{Normalized System Identification Error}
  \bar{e}_{si}(k) &= \tilde{\Phi}(k) \bar{\zeta}({U}_{stf}(k),\alpha (k))-\bar{\omega}(k),\\
  \bar{e}_{si}(k_j) &= \tilde{\Phi}(k) \bar{\zeta}({U}_{stf}(k_j),\alpha (k_j))-\bar{\omega}(k_j),
  \end{aligned}
\end{equation}
where $\bar{\zeta}(k)=\frac{\zeta(k)}{m(k)}$, $\bar{\omega}(k)=\frac{\omega(k)}{m(k)}$, $\bar{\zeta}(k_j)=\frac{\zeta(k_j)}{m(k_j)}$, and $\bar{\omega}(k_j)=\frac{\omega(k_j)}{m(k_j)}$.

\begin{condition} [Rank Condition]
\label{Rank Condition}
The number of linearly independent elements in the history stack $Z$ \eqref{Recorded Data} is the same as the dimension of $\zeta({U}_{stf},\alpha)$; i.e., $rank(Z)=q$.
\end{condition}

Under Condition \ref{Rank Condition}, it is obvious that the following inequalities hold:
\begin{equation}
  \begin{aligned}
    \label{Positive Matrix}
  H_1=\sum\limits_{j=1}^{s} \bar{\zeta}({U}_{stf}(k_j),\alpha (k_j)) \bar{\zeta}^T({U}_{stf}(k_j),\alpha (k_j)) > 0, 
  \end{aligned}
\end{equation}
\begin{equation}
  \begin{aligned}
    \label{Positive Matrix 2}
  H_2=\bar{\zeta}({U}_{stf}(k),\alpha (k)) \bar{\zeta}^T({U}_{stf}(k),\alpha (k)) + H_1 >0. 
  \end{aligned}
\end{equation}

Now, a discrete-time STF-based concurrent learning law is presented in the following theorem.

\begin{theorem} [Discrete-Time STF-based Concurrent Learning]
\label{Discrete-Time Concurrent Learning}
Suppose Condition \ref{Rank Condition} is satisfied. Consider the nonlinear model \eqref{Regressor Format} and the identified model \eqref{Regressor Identified Model}. Then, the discrete-time concurrent learning law
\begin{equation}
  \begin{aligned}
    \label{Update Law}
  &\hat{\Phi}(k+1)=\hat{\Phi}(k) - \bar{e}_{si}(k) \bar{\zeta}^T({U}_{stf}(k),\alpha (k)) \hspace{0.5 mm} \Omega\\
  & \hspace{16 mm} - \sum\limits_{j=1}^{s} \bar{e}_{si}(k_j) \bar{\zeta}^T({U}_{stf}(k_j),\alpha (k_j)) \hspace{0.5 mm} \Omega
  \end{aligned}
\end{equation}
with the learning rate matrix $\Omega=r I_q$,
\begin{equation}
  \begin{aligned}
    \label{Learning rate}
  &0<r<\frac{2\lambda_{min}(H_2)}{\lambda_{max}^2(H_2)}
  \end{aligned}
\end{equation}
guarantees that\\
i) the parameter identification errors $\tilde{\Phi}$ converges to zero when $\omega(k)=0$.\\
ii) the parameter identification errors $\tilde{\Phi}$ is uniformly ultimately bounded (UUB) when $\omega(k) \neq 0$.
\end{theorem}
\begin{proof}
Consider the following Lyapunov function candidate:
\begin{equation}
  \begin{aligned}
    \label{Lyapunov Candidate}
  &V(k) = tr \{\tilde{\Phi}(k) \Omega^{-1} \tilde{\Phi}^T(k)\},  
  \end{aligned}
\end{equation}
where $\tilde{\Phi}(k)$ and $\Omega=r I_q$ are the parameter identification errors and the learning rate matrix defined in \eqref{System Identification Error} and \eqref{Learning rate}, respectively. The derivative of the Lyapunov candidate \eqref{Lyapunov Candidate} in the discrete-time domain is obtained as
\begin{equation}
  \begin{aligned}
    \label{Lyapunov Candidate Derivative}
  &V(k+1)-V(k)\\ 
  &= tr\{\tilde{\Phi}(k+1) \Omega^{-1} \tilde{\Phi}^T(k+1)- \tilde{\Phi}(k) \Omega^{-1} \tilde{\Phi}^T(k)\}\\ 
  &= tr\{(\tilde{\Phi}(k+1)-\tilde{\Phi}(k)) \Omega^{-1} (\tilde{\Phi}(k+1)+\tilde{\Phi}(k))^T\}\\
  &=tr\{(-\bar{e}_{si}(k) \bar{\zeta}^T(k) \hspace{0.5 mm} \Omega -\sum\limits_{j=1}^{s} \bar{e}_{si}(k_j) \bar{\zeta}^T(k_j) \hspace{0.5 mm} \Omega)\Omega^{-1}\\
  &\hspace{4.5 mm}(-\bar{e}_{si}(k) \bar{\zeta}^T(k) \hspace{0.5 mm} \Omega -\sum\limits_{j=1}^{s} \bar{e}_{si}(k_j) \bar{\zeta}^T(k_j) \hspace{0.5 mm} \Omega + 2\tilde{\Phi}(k))^T\}\\
  &=tr\{(-\tilde{\Phi}(k) H_2+H_3) \hspace{0.5 mm} (-\tilde{\Phi}(k) H_2\Omega+H_3\Omega+ 2\tilde{\Phi}(k))^T\},\\
  \end{aligned}
\end{equation}
where $H_3=\bar{\omega}(k)\bar{\zeta}^T(k)+\sum\limits_{j=1}^{s} \bar{\omega}(k_j) \bar{\zeta}^T(k_j)$, and $\|H_3\| \leq \bar{\varepsilon}_n$. Here, one can obtain $\bar{\varepsilon}_n$ using the bound of $\omega$. 
Now, one has
\begin{equation}
  \begin{aligned}
    \label{Lyapunov Candidate Derivative 2}
  &V(k+1)-V(k)\\ 
  &=tr\{\tilde{\Phi}(k) P_1 \tilde{\Phi}^T(k)+\tilde{\Phi}(k) P_2+P_3 \tilde{\Phi}^T(k)+P_4\},\\
  \end{aligned}
\end{equation}
where $P_1=H_2 \Omega^T H^T_2-2H_2$, $P_2=-H_2\Omega^TH^T_3$, $P_3=-H_3 \Omega^T H^T_2+2H_3$, and $P_4=H_3\Omega^TH^T_3$.
Thus, one has
\begin{equation}
  \begin{aligned}
    \label{Lyapunov Candidate Derivative 3}
  &V(k+1)-V(k) \leq Q_1 \|\tilde{\Phi}(k)\|^2 + Q_2 \|\tilde{\Phi}(k)\| + Q_3,\\
  \end{aligned}
\end{equation}
where $Q_1=r \lambda^2_{max}(H_2)-2\lambda_{min}(H_2)$, $Q_2=r\lambda_{min}(H_2)\bar{\varepsilon}_n$ $+r\lambda_{min}(H_2)\bar{\varepsilon}_n+2\bar{\varepsilon}_n$, and $Q_3=r\bar{\varepsilon}^2_n$. Now, using \eqref{Learning rate}, it is clear that $Q_1 < 0$; therefore, using \eqref{Lyapunov Candidate}, one has the following inequality when $\omega(k)=0$:
\begin{equation}
    \label{Lyapunov Candidate Derivative 4}
  V(k+1)-V(k) \leq Q_1 \|\tilde{\Phi}(k)\|^2\leq r \hspace{0.5 mm} Q_1 V(k)< 0.
\end{equation}
Thus, the parameter identification errors $\tilde{\Phi}$ converge to zero. Consequently, from \eqref{System Identification Error}, the convergence of $\tilde{\Phi}$ indicates that the system identification error $e_{si}$ converges to zero. This completes the proof of the first part.

\hspace{-3.55 mm}Now, for $\omega(k) \neq 0$, since $\|\tilde{\Phi}(k)\| \geq 0$, $Q_1 < 0$, $Q_2 \geq 0$, and $Q_3 \geq 0$, the only valid non-negative root of \eqref{Lyapunov Candidate Derivative 3} is
\begin{equation}
  \begin{aligned}
    \label{parameter bound}
  &\tilde{\Phi}_b=\frac{-Q_2-\sqrt{Q^2_2-4 Q_1 Q_3}}{2Q_1}.\\
  \end{aligned}
\end{equation}
Thus, when $\|\tilde{\Phi}(k)\| > \tilde{\Phi}_b$, one has
\begin{equation}
  \begin{aligned}
    \label{Lyapunov Candidate Derivative 5}
  &V(k+1)-V(k) < 0,  
  \end{aligned}
\end{equation}
which makes $\tilde{\Phi}(k)$ to enter and stay in the compact set $S_{\tilde{\Phi}}=\{\tilde{\Phi}: \|\tilde{\Phi}\| \leq \tilde{\Phi}_b\}$; therefore, one can conclude that the system identification error $e_{si}$ converges to a small region around zero using \eqref{System Identification Error}. This completes the proof of the second part.
\end{proof}

\begin{remark} [System Identification]
\label{System Identification}
Theorem 1 presents a learning law to guarantee a reasonable performance for the proposed STF-based concurrent learning. We investigate two cases: i) the real system without any external disturbances, and ii) the real system with external disturbances. For the first case, i.e., $\omega(k) = 0$, we prove that the system identification error converges to zero, and the real system is identified perfectly. For the second case with $\omega(k) \neq 0$, although it is known that the system identification error does not converge to zero, we prove that it is uniformly ultimately bounded (UUB), which means that the system identification error converges to a small region around zero \eqref{parameter bound}. Thus, it is clear that the real system is identified with a reasonable performance even in the presence of external disturbances. To see how we use the learning law in the proposed identifier, we present Algorithm 1 and the corresponding explanations later.
\end{remark}

\begin{remark} [Rich Data]
\label{Rich Data}
The singular value maximizing approach \cite{chowdhary2011singular} is used for recording rich data in the history stack \eqref{Recorded Data}. When Condition \ref{Rank Condition} is satisfied, the richness of the history stack \eqref{Recorded Data} suffices to achieve the results in Theorem \ref{Discrete-Time Concurrent Learning}. However, replacing new rich data with old data in the history stack can improve the learning performance if the new data increases $\lambda_{min}(H_2)$ and/or decreases $\lambda_{max}(H_2)$, leading to a reduced convergence time.
\end{remark}

\begin{remark} [Comparison]
\label{Advantages}
As compared to \cite{chen2020online},  the regressor vector $\zeta({U}_{stf},\alpha)$ has to be persistently exciting (see Definition \ref{def pe}), imposing  conditions on past, current, and future regressor vectors that is difficult or impossible  to verify online. Instead, Condition \ref{Rank Condition} only deals with a subset of past data, which makes it easy to monitor. Also, it is convenient to check whether replacing new data will increase $\lambda_{min}(H_2)$ and/or decrease $\lambda_{max}(H_2)$ to reduce the convergence time. In comparison with \cite{chowdhary2010concurrent}, the STF-based concurrent learning does not require the measurement or estimation of the derivatives of the system states \eqref{system}. In comparison with \cite{vahidi2020memory,vahidi2021learning}, the STF-based concurrent learning does not require a filter regressor for obviating the derivatives of the system states. Compared to our previous work \cite{vahidi2022data}, the discrete-time concurrent learning technique is extended to handle both structured and unstructured uncertainties.
\end{remark}

The steps for transferring from the STF model structure to the state-space form are provided in Appendix~A. For the offline system identification and control policy learning, using Theorem \ref{Discrete-Time Concurrent Learning} and Appendix~A, one can see that $e_{s}$ and $e_{c}$ converge to zero when $\omega(k)=0$ and are UUB when $w \neq 0$. Therefore, after learning process, one can conclude that $e_{s}$ and $e_{c}$ are zero for $w=0$ and around zero for $w \neq 0$.

\begin{figure}[!ht]
     \centering
     \includegraphics[width=1\linewidth]{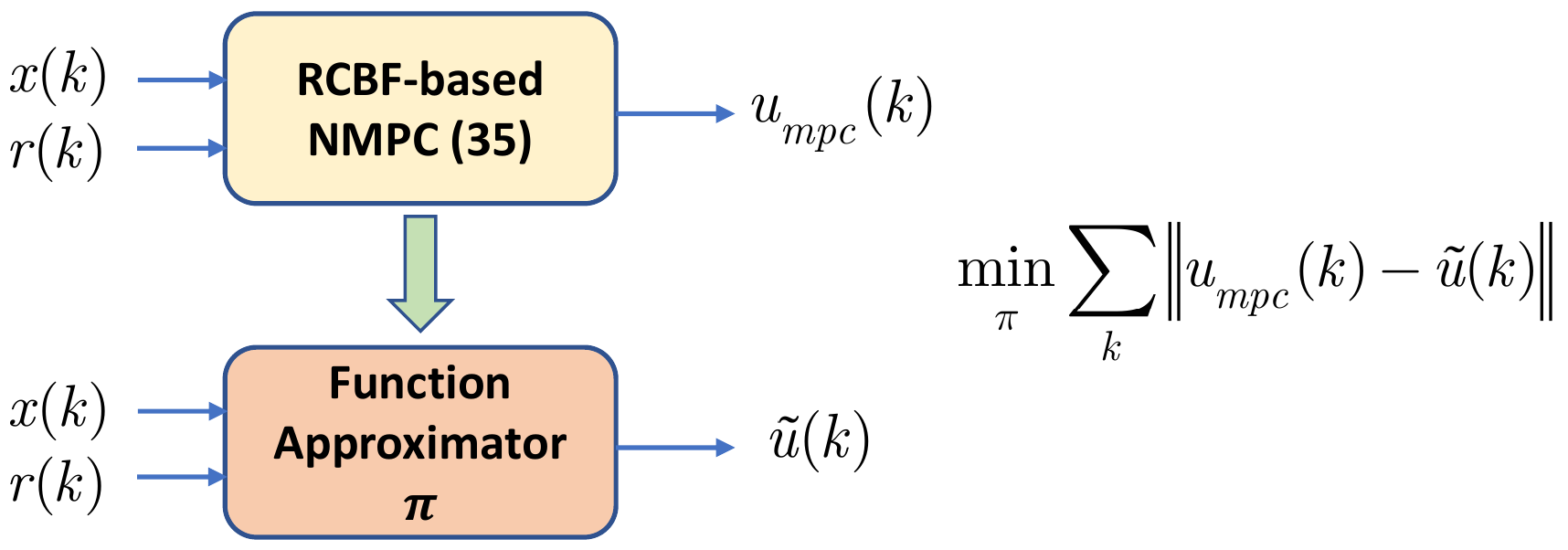}
     \caption{Schematic of NMPC policy function approximation. %A function approximator is used to approximate the NMPC policy \eqref{NMPC} that maps from state $x(k)$, output $y(k)$, and reference $r(k)$ to the control $u(k)$ ($\tilde{u}(k)$).}
     }
     \label{fig:policyLearning}
\end{figure}

\subsection{Data-Driven Safe Predictive Control with Function Approximations}
With the identified STF model, one can use the steps in Appendix~A to transform the STF model to a state-space model and implement the NMPC \eqref{NMPC}. However, in many real-time engineering systems, it is very computationally expensive, especially for nonlinear systems, to implement the NMPC. Therefore, in this section, we develop a data-driven safe predictive control framework using policy function approximators. As illustrated in Fig.~\ref{fig:policyLearning},  we utilize a function approximator (e.g., neural network, Gaussian process regression, or STF) to approximate the NMPC policy, which can be viewed as a mapping from the state $x(k)$ and the reference $r(k)$ to the control $u_{mpc}(k)$. The function approximator can be obtained by minimizing the accumulated loss, i.e., the difference between the NMPC $u_{mpc}$ and the predicted control $\tilde{u}(k)$.

With the trained function approximator, the online control computation is greatly reduced as only simple algebraic computations are needed as compared to the original constrained nonlinear optimization. Due to the control policy learning error, the system identification error, and the unknown disturbances, the constraints in \eqref{Constraints} may not be satisfied. Therefore, we develop a RCBF-based NMPC such that the function-approximated control policy can still guarantee constraint satisfaction in the presence of $e_{c}$, $e_{s}$, and $w$.

Specifically, suppose that the  dynamics drifting \eqref{identification error} due to the system identification error and the control policy learning error are bounded by $\varepsilon_s$ and $\varepsilon_c$, respectively, i.e., $\|e_{s}(k)\| \leq \varepsilon_s$ and $\|e_{c}(k)\| \leq \varepsilon_c$ for all time steps $k$. Note $\varepsilon_s$ and $\varepsilon_c$ can be obtained by empirically evaluating the bounds using trajectories from design of experiments. Consider the state constraints $X=\{x \in \mathbb{R}^n: h(x) \leq 0\}$. We modify the nominal NMPC in \eqref{NMPC} with the following RCBF-based NMPC design:

\begin{equation}
  \begin{aligned}
    \label{RCBF-based NMPC}
    &(\mathbf{x}^{*},\mathbf{u}^{*},\mathbf{\gamma}^{*}) = \underset{\mathbf{\hat{x}},\mathbf{u},\mathbf{\gamma}}{\arg\min} \hspace{1 mm} J_N(\mathbf{\hat{x}},\mathbf{u}) + \varphi(\gamma)\\
    &s.t.\quad \hat{x}(k+1) = \hat{f}(\hat{x}(k),u(k)),\quad \hat{y}(k) = g(\hat{x}(k)),\\
    & \hspace{8.5 mm} \hat{x}(0) = \hat{x}_0,\quad u(k) \in U,\\
    & \hspace{8.5 mm} h(\hat{x}(k+1)) \leq (1-\gamma)h(\hat{x}(k)) -\gamma \eta (\varepsilon_w + \varepsilon_s + \varepsilon_c), \\
    & \hspace{8.5 mm} 0 < \gamma \leq 1,
  \end{aligned}
\end{equation}
where $\varphi(\gamma) = P \gamma^2$ is a regularization term on the optimal variable $\gamma$ with $P > 0$ being a weighting factor. The following theorem shows that approximating the NMPC policy based on the RCBF-based NMPC \eqref{RCBF-based NMPC}, particularly the last constraint, can robustly guarantee constraint satisfaction when using the function-approximated control policy $\tilde{u}(k)$ for online controls.

The following theorem enables us to consider a robust safety constraint using the identified model such that the nonlinear system (1) is safe in the presence of the model uncertainties (i.e., $w, e_{s}, e_{c}$). If the optimization problem (35) is feasible, it is guaranteed to give a control input that satisfies the true CBF constraint (9).

%Considering Assumption \ref{Bounded Terms} and the rich performance of STF for offline system identification and control learning, one can consider $\varepsilon_s,\varepsilon_c > 0$ such that $\|e_{fs}(k)\| \leq \varepsilon_s$ and $\|e_{fc}(k)\| \leq \varepsilon_c$ for all time steps $k$. Now, the safety constraint \eqref{CBF condition2} must be modified to guarantee system safety in the presence of control learning error, system identification error, and external disturbance. Therefore, a discrete-time RCBF is introduced in the following theorem.

\begin{theorem} [Robust Constraint Satisfaction]
\label{Robust Control Barrier Function}
Consider the real system \eqref{system}, the identified model \eqref{nominal model}, the drift dynamics error \eqref{identification error}, and the safe set \eqref{Safe set}. Define $h_r(\hat{x}) = h(\hat{x}) +\eta (\varepsilon_w + \varepsilon_s + \varepsilon_c)$ as the RCBF. The safety constraint
\begin{equation}
  \begin{aligned}
    \label{RCBF Constraint}
  &h(\hat{x}(k+1)) \leq (1-\gamma)h(\hat{x}(k)) -\gamma \eta (\varepsilon_w + \varepsilon_s + \varepsilon_c)  
  \end{aligned}
\end{equation}
guarantees that the safe set $S$ is forward invariant in the presence of control learning error $e_{c}$, system identification error $e_{s}$, and unknown disturbance $w$.
\end{theorem}

\begin{proof}
For notational simplicity, we use $x_k$ to denote $x(k)$ in the following equation. Using the mean value theorem, the identified model \eqref{nominal model}, and the nonlinear system \eqref{disturbed system}, one has
\begin{equation}
  \begin{aligned}
    \label{Mean Value}
  &h(x_k) = h(\hat{x}_k + (x_k-\hat{x}_k)) \\
  & \hspace{8.5 mm}= h(\hat{x}_k) + (h(x_k)-h(\hat{x}_k))\\
  & \hspace{8.5 mm} \leq h(\hat{x}_k) + \left| h(x_k)-h(\hat{x}_k) \right|\\
  & \hspace{8.5 mm} \leq h(\hat{x}_k) + \eta \left\| x_k-\hat{x}_k \right\| \\
  & \hspace{8.5 mm} \leq h(\hat{x}_k) + \eta \left( \left\| w_{k-1} \right\| + \left\| {e_{fs}}_{k-1} \right\| + \left\| {e_{fc}}_{k-1} \right\| \right) \\
  & \hspace{8.5 mm} \leq h(\hat{x}_k) + \eta (\varepsilon_w + \varepsilon_s + \varepsilon_c)\\
  & \hspace{8.5 mm} = h_r(\hat{x}_k).
  \end{aligned}
\end{equation}
Thus, $h_r(\hat{x}_k) \leq 0$ implies that $h(x_k) \leq 0$.\\
Using \eqref{CBF condition}, the following safety constraint guarantees that $h_r(\hat{x}(k)) \leq 0$ for all time steps $k$ as
\begin{equation}
\begin{aligned}
    \label{Robust CBF condition}
  &\triangle h_r(\hat{x}(k)) \leq -\gamma h_r(\hat{x}(k)).
  \end{aligned}
\end{equation}
Now, using \eqref{Mean Value}, one can rewrite \eqref{Robust CBF condition} as
\begin{equation}
\begin{aligned}
    \label{Robust CBF condition2}
  &h(\hat{x}(k+1)) - h(\hat{x}(k)) \leq - \gamma \left( h(\hat{x}(k)) +\eta (\varepsilon_w + \varepsilon_s + \varepsilon_c) \right),
  \end{aligned}
\end{equation}
which is \eqref{RCBF Constraint}, and the proof is completed.
\end{proof}

\begin{remark} [Relaxed RCBF]
\label{Robust Safety}
In the optimization problem \eqref{RCBF-based NMPC}, if $\gamma$ becomes relatively small, the sublevel set of the RCBF $h$ will be smaller, and the system tends to be safer; however, the intersection between the reachable set and the sublevel set might be infeasible. When $\gamma$ becomes larger, the sublevel set will be increased in the state space, which makes the optimization problem more likely to be feasible; however, the RCBF constraint might not be active during the optimization. If $\gamma = 1$, the relaxed RCBF constraint becomes equivalent to a simple distance constraint; consequently, the NMPC needs a longer horizon of prediction to generate an expected system safety performance in the closed-loop trajectory. Therefore, it is not recommended to set a relatively too small value for $P$, which would over-relax the RCBF constraint and make the optimized value $\gamma$ closer to 1.
\end{remark}

\begin{comment}
Now, considering the robust safety constraint \eqref{RCBF Constraint} for the state constraints in \eqref{Constraints}, $\gamma$ is defined as an optimal variable such that $0<\gamma \leq 1$. Thus, the discrete-time RCBF-based NMPC is developed as
\begin{equation}
  \begin{aligned}
    \label{RCBF-based NMPC}
    &(\mathbf{x}^{*},\mathbf{u}^{*},\mathbf{\gamma}^{*}) = \underset{\mathbf{\hat{x}},\mathbf{u},\mathbf{\gamma}}{\arg\min} \hspace{1 mm} J_N(\mathbf{\hat{x}},\mathbf{u}) + \varphi(\gamma)\\
    &s.t.\quad \hat{x}(k+1) = \hat{f}(\hat{x}(k),u(k)),\quad \hat{y}(k) = g(\hat{x}(k))\\
    & \hspace{8.5 mm} \hat{x}(0) = x(k),\quad u(k) \in U, \quad \hat{x}(k) \in X,\\
    & \hspace{8.5 mm} h(\hat{x}(k+1)) \leq (1-\gamma)h(\hat{x}(k)) -\gamma \eta (\varepsilon_w + \varepsilon_s + \varepsilon_c),
  \end{aligned}
\end{equation}
where $\varphi(\gamma) = P \gamma^2$ is an additional cost for the optimal variable $\gamma$ with $P > 0$ being a weighting factor. Moreover, the feasibility of an optimization problem in the presence of system constraints are referred to \cite{zeng2021enhancing,zeng2021safetyb}.
\end{comment}

\begin{remark} [Feasibility]
\label{Feasibility}
The recursive feasibility is generally not guaranteed for the MPC in the presence of both safety constraint and control input constraint \cite{borrelli2017predictive,lars2011nonlinear}. In \cite{taylor2020adaptive}, the control input constraint is relaxed to ensure the feasibility of the optimization problem with the safety constraint. However, in this paper, the decay rate of the RCBF is relaxed from a fixed value to an optimal variable so that it improves both safety and feasibility of the NMPC problem \eqref{RCBF-based NMPC}. Formal guarantee on the recursive feasibility in the presence of both safety and control input constraints requires more analysis and will be addressed in our future work. It is worth noting that even in the absence of control input constraint, the optimization problem \eqref{RCBF-based NMPC} will be infeasible if the model uncertainties in the RCBF constraint are too high. The optimization problem \eqref{RCBF-based NMPC} is formulated with an identified model, and the proposed offline system identification technique is utilized to well capture the system dynamics, ensuring that the model uncertainties will be varying within bounded compatible sets online (i.e., $\varepsilon_w, \varepsilon_s, \varepsilon_c$). In particular, under the assumption that the external disturbance $w$ is bounded (see Assumption 1), our STF-based concurrent learning ensures that the bounds on the system identification error and the control learning error are zero and small values close to zero for $w = 0$ and $w \neq 0$, respectively (see Theorem 1 and Remark 1). Therefore, thanks to Theorem 1, the RCBF constraint \eqref{RCBF Constraint} will be satisfied, and the optimization problem \eqref{RCBF-based NMPC} without control input constraints will be feasible at all times. %In addition, for the online learning case, the feasibility of the optimization problem \eqref{RCBF-based NMPC} can be guaranteed with high probability using the prepared analysis in \cite{castaneda2022probabilistic}.
\end{remark}

Now, the proposed discrete-time STF-based concurrent learning is used again to learn the RCBF-based NMPC policy \eqref{RCBF-based NMPC}, which $(x(k),r(k:k+N))$ and $u_{mpc}(k)$ are respectively considered as the input and the output of the policy. %Consequently, $U_{stf}(k) = \left[ x(k), r(k), ..., r(k+N) \right]$ is considered as the input of the STF function approximator. %However, if we apply the offline data-driven safe predictive control to the real system \eqref{system}, one may have performance loss due to $e_{c}$, $e_{s}$, and $w$. Thus, in the following subsection, an online correction scheme is proposed to minimize the performance loss for the real system.

\subsection{Online Adaptive Control Policy}
Although the function-approximated control policy $\tilde{u}(k)$ may already lead to a reasonable closed-loop performance for many cases, one can see a performance loss for a system due to i) the choice of the hyperparameters that decide the architecture of the STF-based identification law and the function-approximated control policy, and ii) insufficient training data in some regions of the feasible state space. To minimize this performance loss, caused by $e_{c}$, $e_{s}$, and $w$, it may be desirable to adapt $\tilde{u}(k)$ online. %We propose an online correction scheme, including an KKT adaptation and an ancillary feedback control, to handle this issue.

%The following corollary is used to propose the KKT adaptation and minimize the asymptotic performance loss caused by the control learning error $e_{c}$.

\begin{corollary}[Steady-State Optimization Problem \cite{krishnamoorthy2021adaptive}]
\label{lemm1}
Using the RCBF-based NMPC policy, a unique equilibrium pair $(x_e,u_e)$ minimizes the steady-state optimization problem
\begin{equation}
  \begin{aligned}
    \label{Steady-State NMPC}
    &(\mathbf{x}_{e},\mathbf{u}_{e},\mathbf{\gamma}_{e}) = \underset{\mathbf{\hat{x}},\mathbf{u},\mathbf{\gamma}}{\arg\min} \hspace{1 mm} l(\hat{x},u,\gamma)\\
    &s.t.\quad \hat{x} = \hat{f}(\hat{x},u),\quad \hat{y} = g(\hat{x})\\
    & \hspace{8.5 mm} u \in U,\\
    & \hspace{8.5 mm} h(\hat{x}) \leq (1-\gamma)h(\hat{x}) -\gamma \eta (\varepsilon_w + \varepsilon_s + \varepsilon_c),
  \end{aligned}
\end{equation}
where $l(\hat{x},u,\gamma) = \phi(\hat{x},u,\hat{y},r) + \psi(\hat{x},\hat{y},r) + \varphi(\gamma)$ represents the steady-state cost function. \hspace{42 mm} $\square$
\end{corollary}

\begin{condition}[Steady-State Function-Approximated Control]
\label{cond1}
$(x^{\prime}_e,u^{\prime}_e)$ is the asymptotically stable equilibrium point of the nominal model \eqref{nominal model} under the approximated control policy $\tilde{u}$.
\end{condition}

\begin{lemma}[Modified Steady-State Optimization Problem]
\label{Modified Steady-State Optimization Problem}
For the equilibrium point $(x_e,u_e)$, the steady-state optimization problem \eqref{Steady-State NMPC} can implicitly be written as
\begin{equation}
  \begin{aligned}
    \label{Modified Steady-State NMPC}
    &\mathbf{u}_{e} = \underset{\mathbf{u}}{\arg\min} \hspace{1 mm} \tilde{l}(u)\\
    &s.t.\quad \tilde{h}(u)\leq 0,
  \end{aligned}
\end{equation}
where the constraints of the optimization problem \eqref{Steady-State NMPC} are collectively denoted as $\tilde{h}(u)$.
\end{lemma}

\begin{proof}
Let $\hat{x}(k)$ and $\tilde{u}(k)$ denote the nominal trajectory using the function-approxiamted control policy. Thus, the nominal trajectory for the RCBF-based NMPC policy \eqref{RCBF-based NMPC} is expressed as
\begin{equation}
  \begin{aligned}
    \label{Difference}
    & \hat{x}_{mpc}(k)=\hat{x}(k)+\delta \hat{x}(k),\\
    & u_{mpc}(k)=\tilde{u}(k)+\delta u(k),
  \end{aligned}
\end{equation}
where $\delta u(k)$ represents the difference between the NMPC and function-approximated control policy at each time step $k$, and $\delta \hat{x}(k)$ denotes the resulting change in the states. Using the nominal model \eqref{nominal model}, one has
\begin{equation}
  \begin{aligned}
    \label{Variation}
    & \delta \hat{x}(k+1)= \hat{f}_{\hat{x}} (k) \delta \hat{x}(k) + \hat{f}_{u} (k) \delta u(k),\\
    & \delta l(k)= l_{\hat{x}}(k) \delta \hat{x}(k) + l_{u} (k) \delta u(k),
  \end{aligned}
\end{equation}
where $\hat{f}_{\hat{x}}(k)$, $\hat{f}_{u}(k)$, $l_{\hat{x}}(k)$, and $l_{u}(k)$ are partial derivatives with respect to $\hat{x}$ and $u$, respectively, and  obtained using $\hat{x}(k)$ and $\tilde{u}(k)$.

\hspace{-5 mm} Now, it is clear that $\delta \hat{x}(k+1)=\delta \hat{x}(k)$ for the steady-state condition; therefore, one has
\begin{equation}
  \begin{aligned}
    \label{eq13}
    \delta \hat{x}(k)= (I_n-\hat{f}_{\hat{x}}(k))^{-1} \hat{f}_{u}(k) \hspace{0.75 mm} \delta u(k),
  \end{aligned}
\end{equation}
and
\begin{equation}
  \begin{aligned}
    \label{eq14}
    \delta l(k)= (l_{\hat{x}}(k) (I_n-\hat{f}_{\hat{x}}(k))^{-1} \hat{f}_{u}(k) + l_{u} (k)) \hspace{0.75 mm} \delta u(k).
  \end{aligned}
\end{equation}
Using \eqref{eq14}, one has $\delta \tilde{l}(u) = (l_{\hat{x}} (I_n-\hat{f}_{\hat{x}})^{-1} \hat{f}_{u} + l_{u})$ and can obtain $\tilde{l}(u)$. Similarly, $\tilde{h}(u)$ is obtained using the same process. This completes the proof.
\end{proof}

Now, consider $\tilde{x}=\hat{x}-x$ and define an auxiliary variable $s$ as
\begin{equation}
  \begin{aligned}
    \label{s}
    s(k)=\Gamma \tilde{x}(k),
  \end{aligned}
\end{equation}
where $\Gamma \in \mathbb{R}^{m \times n}$ is designed such that $\Gamma \hat{f}_u \in  \mathbb{R}^{m \times m}$ is a diagonal matrix. The following theorem presents the proposed online correction scheme, including an KKT adaptation and an ancillary feedback control.

\begin{theorem}[Online Adaptive Control Policy]
\label{Online Adaptive Control Policy}
Considering the real system \eqref{system}, the nominal model \eqref{nominal model}, Corollary \ref{Modified Steady-State Optimization Problem}, Lemma \ref{Modified Steady-State Optimization Problem}, and the auxiliary variable $s$ \eqref{s}, the online adaptation policy
\begin{equation}
  \begin{aligned}
    \label{Online Adaptation}
    & u(k)=\tilde{u}(k)+\delta u(k)+K(\hat{x}(k),x(k)),\\
    & \delta u(k) \approx \delta u(k-1)-K_0 \begin{bmatrix} \tilde{h}_a(u(k-1))\\ N^T \delta \tilde{l}(u(k-1)) \end{bmatrix},\\
    & K_0=\begin{bmatrix} \delta \tilde{h}_a(u(k-1))\\ N^T \delta^2 \tilde{l}(u(k-1)) \end{bmatrix}^{-1},\\
    & K(\hat{x}(k),x(k))=\frac {1}{\Gamma \hat{f}_{u}(\hat{x}(k),\tilde{u}_{mpc}(k))} (\Upsilon s(k)+\Gamma \varepsilon_w+\Gamma \varepsilon_s).
  \end{aligned}
\end{equation}
with the design matrix $\Upsilon \in \mathbb{R}^{m \times m}$ $(0 \leq \Upsilon_{ii}<1, \hspace{1 mm} i=1,...,m)$ and $\tilde{u}_{mpc}(k) = \tilde{u}(k)+\delta u(k)$ minimizes the performance loss due to the control learning error, the system identification error, and the unknown disturbance.
\end{theorem}

\begin{proof}
The first part of the proof focuses on minimizing the performance loss for the nominal model \eqref{nominal model} due to the control learning error. Considering Corollary \ref{Modified Steady-State Optimization Problem} and Lemma \ref{Modified Steady-State Optimization Problem}, one can conclude that $(x_e,u_e)$ is an asymptotically stable equilibrium point of the nominal model under the RCBF-based NMPC policy. Hence, $u_{mpc}$ satisfies the KKT conditions of \eqref{Modified Steady-State NMPC}, which are expressed as
\begin{equation}
  \begin{aligned}
    \label{KKT}
    & \delta \tilde{l}(u) + \delta \tilde{h}_a(u)^T \lambda=0,\\
    & \lambda^T \tilde{h}(u)=0,
  \end{aligned}
\end{equation}
where $\tilde{h}_a(u)$ denotes the active constraints, and $\lambda\geq 0$ is the Lagrange multiplier. Now, one can rewrite \eqref{KKT} as
\begin{equation}
  \begin{aligned}
    \label{KKT 2}
    & \mathcal{A}^T \delta \tilde{l}(u)=0,\\
    & \tilde{h}_a(u)=0,
  \end{aligned}
\end{equation}
where $\mathcal{A}$ lies in the null space of the active constraint variation, i.e., $\mathcal{A}^T \delta \tilde{h}_a(u)^T=0$.

\hspace{-5 mm} Due to the control learning error $e_{c}$, one may have $(x^{\prime}_e,u^{\prime}_e) \neq (x_e,u_e)$; therefore, the goal is that the function-approximated control policy $\tilde{u}(k)$ is adapted online such that it guarantees $(x^{\prime}_e,u^{\prime}_e)=(x_e,u_e)$. Using Lemma \ref{lemm1}, it is clear that $(x^{\prime}_e,u^{\prime}_e)$ does not satisfy the KKT conditions \eqref{KKT 2} if $(x^{\prime}_e,u^{\prime}_e) \neq (x_e,u_e)$. Thus, the deviation from the KKT condition \eqref{KKT 2} indicates asymptotic performance loss stemming from the RCBF-based NMPC policy approximation. Consequently, $\delta u$ minimizes the asymptotic performance loss due to $e_{c}$, and the gain $K_0$ is chosen such that it does not significantly affect the dynamics, but adjusts the asymptotic performance. Therefore, the first part of the proof is completed.

\hspace{-5 mm} Now, we have a reasonable performance for the nominal model under the function-approximated control policy $\tilde{u}(k)$; however, the performance loss still exists for the real system \eqref{system} due to the system identification error and the unknown disturbance. Considering the nominal model \eqref{nominal model} under $\tilde{u}_{mpc}$, one has
\begin{equation}
  \begin{aligned}
    \label{x tilde}
    &\tilde{x}(k+1)=\hat{x}(k+1)-x(k+1)\\
    &=\hat{f}(\hat{x}(k),\tilde{u}_{mpc}(k))-f(x(k),u(k))-w(k)\\
    &=\hat{f}(\hat{x}(k),\tilde{u}_{mpc}(k))-\hat{f}(\hat{x}(k),u(k))+e_{s}(k)-w(k)\\
    &= \hat{f}_{u}(\hat{x}(k),\tilde{u}_{mpc}(k)) \hspace{0.7 mm} (\tilde{u}_{mpc}(k)-u(k)) + e_{s}(k) - w(k)\\
    &=- \hat{f}_{u}(\hat{x}(k),\tilde{u}_{mpc}(k)) \hspace{0.7 mm} K(\hat{x}(k),x(k)) + e_{s}(k) - w(k).
  \end{aligned}
\end{equation}
Now, the Lyapunov function candidate is considered as
\begin{equation}
  \begin{aligned}
    \label{Lyapunov}
    V(k)=s(k)^T s(k),
  \end{aligned}
\end{equation}
where $V(k)$ is a positive definite function. The derivative of the Lyapunov candidate \eqref{Lyapunov} in the discrete-time domain is obtained as
\begin{equation}
  \begin{aligned}
    \label{Lyapunov derivative}
    V(k+1)-V(k)=s(k+1)^T s(k+1)-s(k)^T s(k).
  \end{aligned}
\end{equation}
Using \eqref{s}-\eqref{x tilde}, one has
\begin{equation}
  \begin{aligned}
    \label{s 2}
    &s(k+1)=\Gamma \tilde{x}(k+1)\\
    &=-\Upsilon s(k)-\Gamma \varepsilon_w-\Gamma \varepsilon_s+\Gamma e_f(k) -\Gamma w(k)\\
    &\leq -\Upsilon s(k)-\Gamma \varepsilon_w-\Gamma \varepsilon_s+\Gamma \|e_f(k)\| +\Gamma \|w(k)\|\\
    &\leq -\Upsilon s(k).
  \end{aligned}
\end{equation}
Using \eqref{Lyapunov derivative} and \eqref{s 2}, considering $P=I_m - \Upsilon^T \Upsilon$, one has
\begin{equation}
  \begin{aligned}
    \label{Lyapunov derivative 2}
    &V(k+1)-V(k) \leq -s(k)^T P s(k)\\
    &\hspace{25.4 mm} \leq -\lambda_{min}(P)V(k)\\
    &\hspace{25.4 mm} <0.
  \end{aligned}
\end{equation}
Consequently, one can conclude that $\delta u(k)$ minimizes the difference between $u_{mpc}$ and $\tilde{u}$, and $K(\hat{x},x)$ keeps the actual states $x$ of the real system around the nominal trajectory $\hat{x}$ under $\tilde{u}_{mpc}$; therefore, $K(\hat{x},x)$ minimizes the performance loss due to system identification error $e_{s}$ and the unkknown disturbance $w$. This completes the proof.
\end{proof}

\begin{remark} [KKT Condition]
\label{KKT Condition}
To develop the proposed adaptation law \eqref{Online Adaptation}, Condition 2 must be satisfied for the nominal model, which means that its closed-loop performance under the function-approximated control policy $\tilde{u}(k)$ converges to an equilibrium point $(x^{\prime}_e,u^{\prime}_e)$, but it may not be the desired equilibrium point $(x_e,u_e)$ from the RCBF-based NMPC policy $u_{mpc}(k)$ because of the policy approximation error. When training $\tilde{u}(k)$ using the RCBF-based NMPC policy $u_{mpc}(k)$, one wants the equilibrium solution to converge to the same limit point $(x_e,u_e)$, where the KKT conditions hold. Consequently, we adapt $\tilde{u}(k)$ online to ensure that $(x^{\prime}_e,u^{\prime}_e)$ satisfies the KKT conditions and make $(x^{\prime}_e,u^{\prime}_e) = (x_e,u_e)$.
\end{remark}

\begin{remark} [Offline Probabilistic Verification]
\label{Offline Probabilistic Verification}
To guarantee that Condition \ref{cond1} is satisfied, using the offline data-driven safe predictive control, the nominal model is simulated for $N_v$ randomly selected initial states in the training data range as
\begin{equation}
  \begin{aligned}
    \label{eq20}
    N_v \geq \frac{\log \frac {1}{\kappa }}{\log \frac {1}{1-\epsilon}},
  \end{aligned}
\end{equation}
where $\epsilon \in (0,1)$ and $\kappa  \in (0,1)$ denote the accuracy and confidence of the offline probabilistic verification. If all $N_v$ closed-loop trajectories are stable ($E=0$), one can conclude that the nominal model under the function-approximated control policy converges to $(x^{\prime}_e,u^{\prime}_e)$ for all initial states with the probability \cite{krishnamoorthy2021adaptive}
\begin{equation}
  \begin{aligned}
    \label{eq21}
    Pr\{Pr\{E=0\} \geq 1-\epsilon\} \geq 1-\kappa.
  \end{aligned}
\end{equation}
However, if any closed-loop trajectory from $N_v$ samples is unstable, the control learning procedure must be repeated.
\end{remark}

\begin{remark} [Safety and Performance]
\label{Control Learning Error}
Using Theorem \ref{Online Adaptive Control Policy}, it is clear that the online adaptive control policy \eqref{Online Adaptation} keeps the actual states $x$ around the nominal trajectory $\hat{x}$ under $\tilde{u}_{mpc}$. Therefore, using the fourth line in \eqref{Mean Value}, one can make sure that the online adaptive control policy improves the safety for the real system. Moreover, since the performance loss caused by $e_{c}$, $e_{s}$, and $w$ is mitigated, the online adaptive control policy \eqref{Online Adaptation} has better performance than the function-approximated control policy $\tilde{u}(k)$ and the RCBF-based NMPC policy \eqref{RCBF-based NMPC} for online controls.
\end{remark}

\begin{figure}[!ht]
     \centering
     \includegraphics[width=0.99\linewidth]{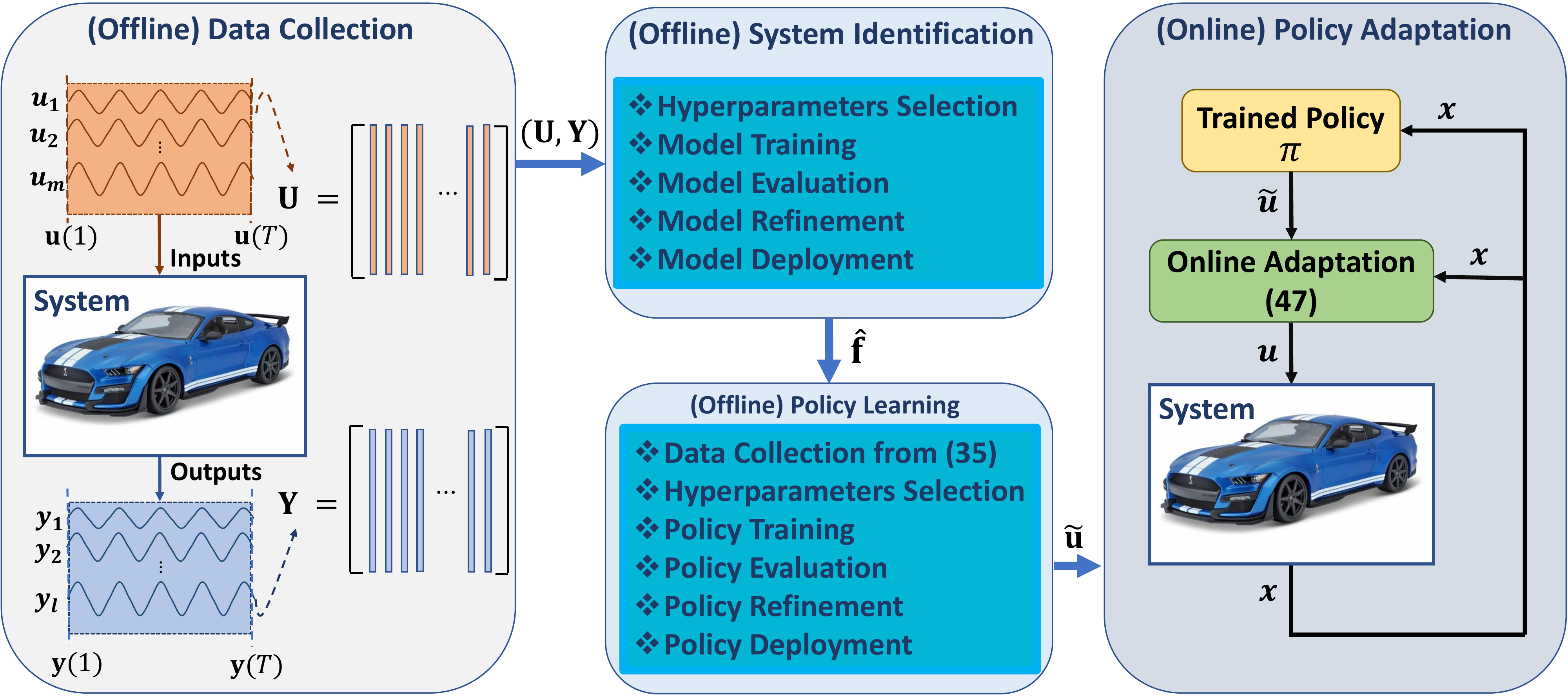}
     \caption{Sequence of Steps of online data-driven safe predictive control.}
     \label{Flowchart}
\end{figure}

\begin{figure}[!ht]
     \centering
     \includegraphics[width=0.99\linewidth]{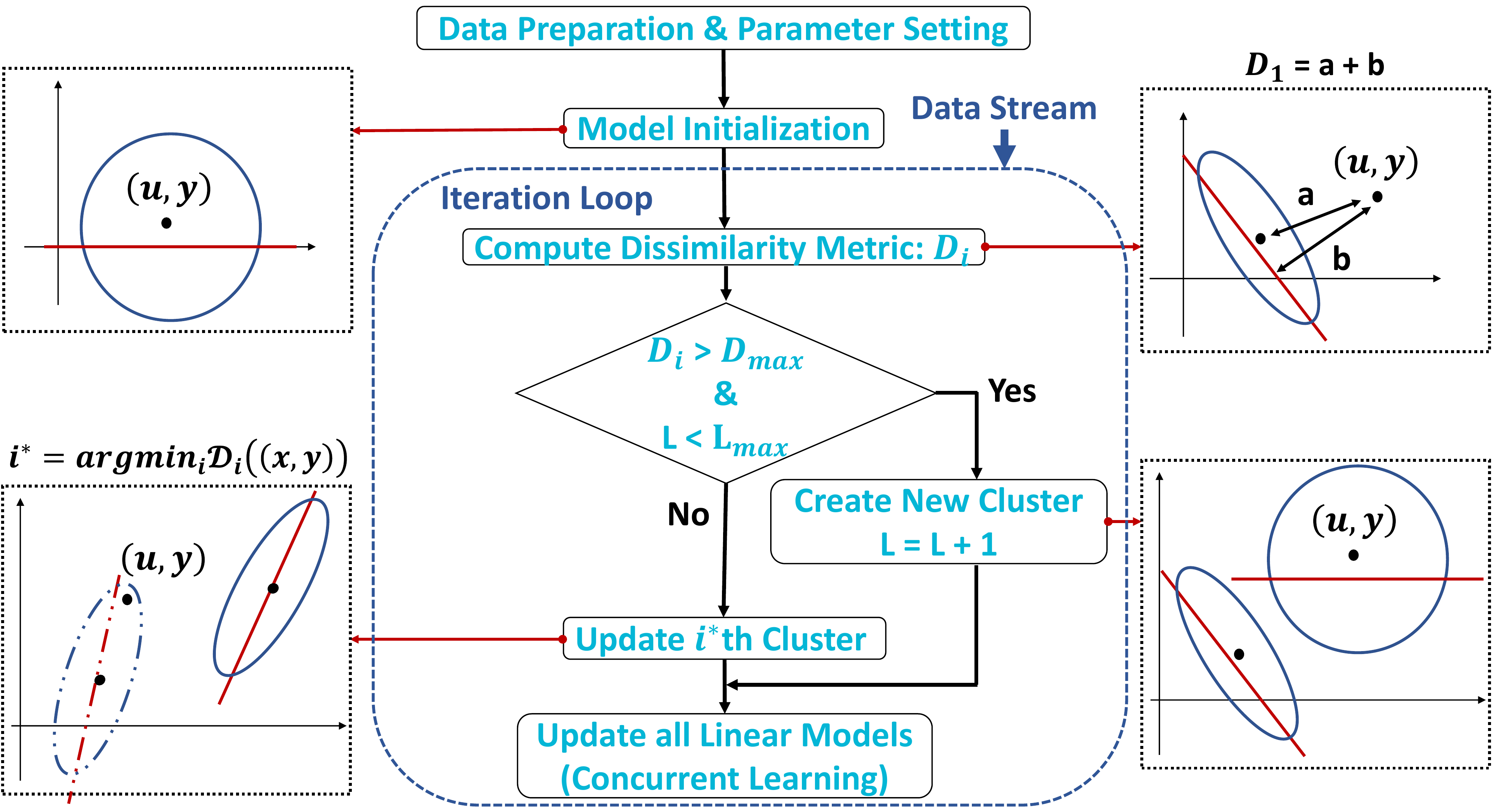}
     \caption{Sequence of Steps of discrete-time STF-based Concurrent Learning: i) First cluster-model is initialized for the first I/O data, ii) For next data, the dissimilarity metric is computed to either create a new cluster-model or update one of the existing ones, iii) For each iteration, all linear models are updated using the CL, and iv) The process is repeated until all training I/O data is used. }
     \label{STF Flowchart}
\end{figure}

\begin{remark} [Computational Cost]
\label{Computational Cost}
The online data-driven safe predictive control is much less computationally costly than the RCBF-based NMPC problem \eqref{RCBF-based NMPC} since the proposed control policy learns the NMPC policy while keeping the real system safe; however, it does not need to solve an optimization problem at each time step $k$ (only algebraic computations are needed). In comparison with our previous work \cite{vahidi2022data}, the RCBF is extended to guarantee system safety in the presence of not only the system identification error and the external disturbance but also the control learning error. Therefore, we have removed the QP safety filter from the algorithm. Moreover, the proposed online adaptive control policy enables us to minimize the performance loss; thus, we do not need to consider a switching criteria for returning the NMPC. These two contributions effectively reduce the computational cost for the proposed algorithm.
\end{remark}

The procedure for determining the online data-driven safe predictive control is summarized in Fig. \ref{Flowchart} and Algorithm \ref{STF}. Fig. \ref{Flowchart} shows a flowchart that represents the sequence of the steps of the proposed data-driven control framework, and Algorithm 1 describes each step in the flowchart. Moreover, Fig. \ref{STF Flowchart} presents a flowchart that describes the step of model training using the discrete-time STF-based concurrent learning, and the readers are referred to see Algorithms 1 and 2 in \cite{chen2020online} for more details about the sequence of the steps of the STF framework. It is worth noting that we have modified the Algorithm 2 in \cite{chen2020online} such that we use the concurrent learning law instead of the recursive least squares (RLS) law to remove the PE requirement.

\begin{algorithm}[ht]
    \caption{Online data-driven safe predictive control}
    \label{STF}
    \textbf{I. Offline System Identification}: \\
    1- \textbf{Data Collection:} Collect the raw input/output data $(u/y)$ from the real system. \\
    2- \textbf{Data Processing:} Process the data to prepare it for model learning, which includes data normalization and data partitioning to training data and validation data. \\
    3- \textbf{Hyperparameters Selection:} Determine the appropriate STFs' hyperparameters to use for model learning, which can be done through trial and error, feature selection, or genetic algorithm (GA). \\
    4- \textbf{Model Training:} Train a model on the processed data using the discrete-time STF-based concurrent learning on the training data to identify the real system and obtain the identified model \eqref{nominal model}. \\
    5- \textbf{Model Evaluation:} Evaluate the performance of the trained model using the best fitting rate, which analyzes the goodness of fit between the measured output and the simulated output of the trained model on the validation data. \\
    6- \textbf{Model Refinement:} If the performance of the trained model is not satisfactory, refine the model by adjusting the hyperparameters. \\
    7- \textbf{Model Deployment:} Once the model is trained and validated, it can be deployed for real-world use. This may involve retraining the trained model with new data, which may need updating STFs' hyperparameters. \\
    \textbf{II. Offline NMPC Learning}: Using the identified model \eqref{nominal model}, we collect raw data from the RCBF-based NMPC \eqref{RCBF-based NMPC} such that $x,r$ and $u$ are considered as the input and the output of the control policy. Doing same procedure as the step I (i.e., offline system identification), we use the discrete-time STF-based concurrent learning as a function approximator to learn the RCBF-based NMPC policy. Once the policy is trained and validated, it can be deployed for real-world use. This may involve retraining the trained policy with new data, which may need updating STFs' hyperparameters. \\ 
    \textbf{III. Online Adaptation}: The  adaptation law \eqref{Online Adaptation} is used online to minimize the performance loss caused by the control learning error, the system identification error, and the unknown disturbance.
\end{algorithm}

\section{Simulation Results}
\label{Sec4}
In this section, two simulation examples are considered to demonstrate the effectiveness of the proposed online data-driven safe predictive control. The simulation results are presented for a cart-inverted pendulum and a gasoline engine controls. The cart-inverted pendulum has a known model, where we use it for NMPC to compare its result with the proposed online data-driven safe predictive control. However, there is not any known model for the gasoline engine vehicle; thus, we use the collected data from the system to identify a nominal model for the NMPC and compare its result with the proposed online data-driven safe predictive control.

\begin{figure}[!ht]
     \centering
     \includegraphics[width=0.59\linewidth]{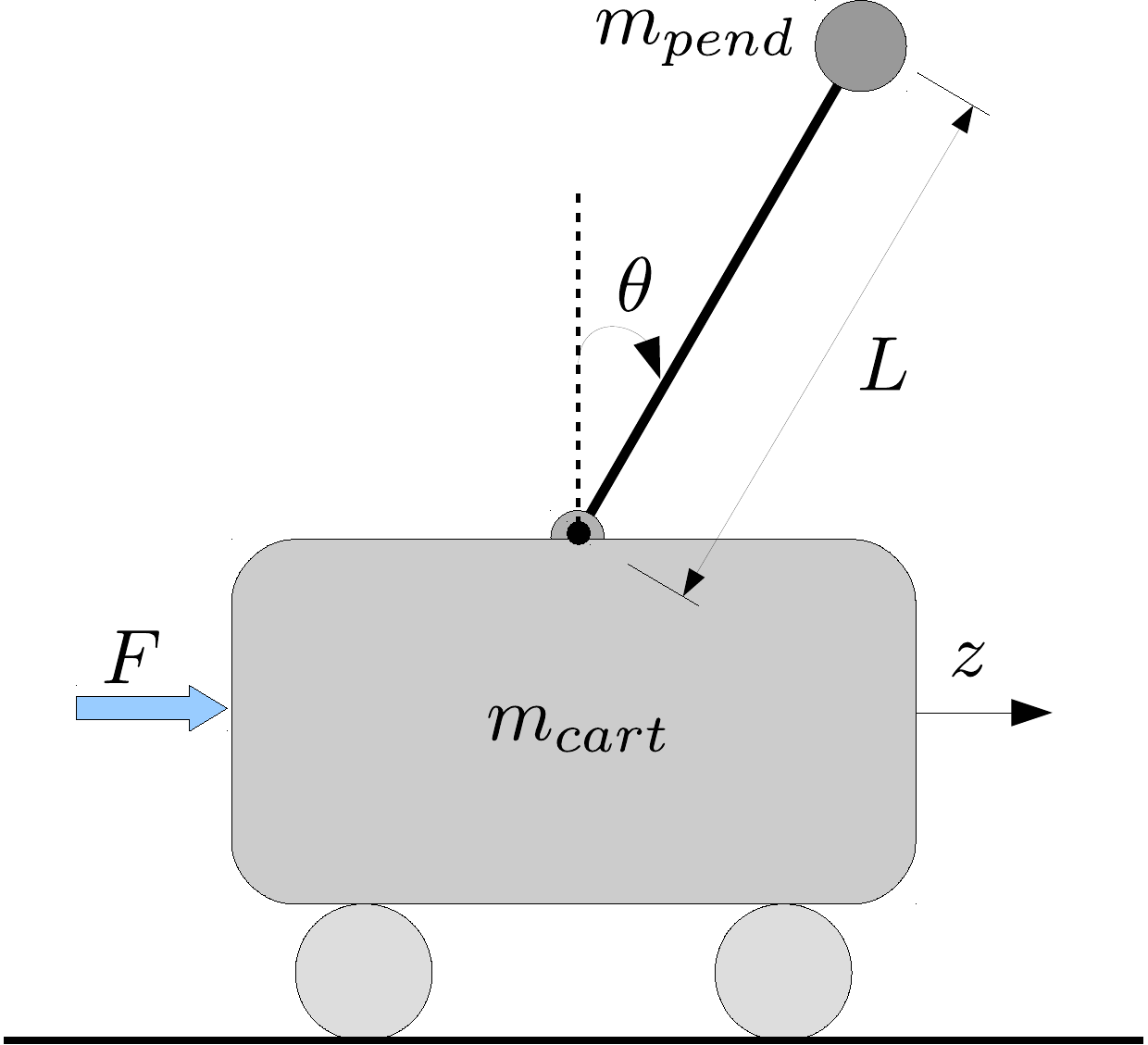}
    \caption{Cart-inverted pendulum.}
     \label{Cart-inverted pendulum fig}
 \end{figure}

\subsection{Cart-Inverted Pendulum}
As shown in Fig. \ref{Cart-inverted pendulum fig}, an inverted pendulum mounted to a cart is modeled in continuous-time domain as
\begin{equation}
  \begin{aligned}
    \label{Cart-inverted pendulum}
    &\ddot{z}=\frac{F-{K}_{d}\dot{z}-{m}_{pend}(L{\dot{\theta }}^{2} \sin (\theta )-g\sin (\theta ) \cos (\theta ))}{{m}_{cart}+{m}_{pend}{\sin }^{2}(\theta )},\\ 
    &\ddot{\theta }=\frac{\ddot{z}\cos (\theta )+g\sin (\theta )}{L},
  \end{aligned}
\end{equation}
where $z$ and $\theta$ are the cart position and the pendulum angle, respectively. ${m}_{cart}=5kg$, ${m}_{pend}=1kg$, $L=2m$, ${K}_{d}=10Ns/m$, and $g=9.81m/s^{2}$ present the cart mass, the pendulum mass, the length of the pendulum, the damping parameter, and the gravity acceleration, respectively. The system is controlled by a variable force $F$, and the model \eqref{Cart-inverted pendulum} is discretized with a sampling time $T=0.1s$.

Now, the states and the outputs of the model \eqref{Cart-inverted pendulum}, the safety constraint, and the input constraint are considered as
\begin{equation}
  \begin{aligned}
    \label{pendulum's states}
    	x={[{x}_{1},{x}_{2},{x}_{3},{x}_{4}]}^{T}={[z,\dot{z},\theta,\dot{\theta }]}^{T},
  \end{aligned}
\end{equation}
\begin{equation}
  \begin{aligned}
    \label{pendulum's output}
    	y={[{x}_{1},{x}_{3}]}^{T}={[z,\theta]}^{T},
  \end{aligned}
\end{equation}
\begin{equation}
  \begin{aligned}
    \label{pendulum's safety}
    	-5\le z\le 5,
  \end{aligned}
\end{equation}
\begin{equation}
  \begin{aligned}
    \label{pendulum's constraints}
    	-100\le F\le 100.
  \end{aligned}
\end{equation}

Using \eqref{Safe set}, \eqref{Mean Value}, and \eqref{pendulum's safety}, the CBF and the RCBF are considered as
\begin{equation}
  \begin{aligned}
    \label{pendulum's CBF}
    	h(x)=\left\| x \right\|-0.5,
  \end{aligned}
\end{equation}
\begin{equation}
  \begin{aligned}
    \label{pendulum's RCBF}
    	{h}_{r}(\hat{x})=\left\| \hat{x} \right\|-0.5+\eta (\varepsilon_w + \varepsilon_s + \varepsilon_c),
  \end{aligned}
\end{equation}
where $\eta=1$, $\varepsilon_w=0.01$, $\varepsilon_s=0.01$, and $\varepsilon_c=0.01$ are considered in \eqref{pendulum's RCBF}. Using \eqref{pendulum's states} and \eqref{pendulum's safety}, it is clear that the RCBF is only considered for $x_1$.

According to Algorithm \ref{STF}, the online data-driven safe predictive control is applied to the discrete-time version of the cart-inverted pendulum \eqref{Cart-inverted pendulum}, which yields 1) 97.54\% accuracy for 30000 training data with 10 clusters and local linear models in the offline system identification, and 2) 98.66\% accuracy for 100000 training data with 20 clusters and local linear models in the offline RCBF-based NMPC learning.

For the proposed online data-driven safe predictive control, the nonlinear system \eqref{Cart-inverted pendulum} is considered in the presence of disturbance ${w}(k)=-0.01+0.02\times rand(4,1)$. Fig. \ref{Cart-inverted pendulum Control input} depicts the control input signal $F$ using the RCBF-based NMPC, the offline data-driven safe predictive control, and the online data-driven safe predictive control with four future state predictions $N=4$ for the cart-inverted pendulum. For this example, it should be mentioned that the RCBF-based NMPC policy \eqref{RCBF-based NMPC} works with the well-known model \eqref{Cart-inverted pendulum}; however, we adapt it using the feedback controller in \eqref{Online Adaptation} to minimize the performance loss caused by the considered external disturbance ${w}(k)$. Moreover, one can see that the proposed online adaptive control policy corrects the offline data-driven safe predictive control and approximates the RCBF-based NMPC policy better. Fig. \ref{Cart-inverted pendulum Outputs} shows that the outputs of the cart-inverted pendulum \eqref{Cart-inverted pendulum} converge to zero using three control schemes whereas one can see the performance loss if only the offline data-driven safe predictive control is used. From Figs. \ref{Cart-inverted pendulum Control input} and \ref{Cart-inverted pendulum Outputs}, one can see that Condition 2 is satisfied such that the system converges to an equilibrium point $(x^{\prime}_e,u^{\prime}_e)$ using the offline data-driven safe predictive control; however, it is not same as the equilibrium point $(x_e,u_e)$ obtained by the RCBF-based NMPC policy. Therefore, the role of the online adaptive control policy is clearly demonstrated so that it makes same equilibrium point $(x_e,u_e)$ for the real system. Furthermore, Figs. \ref{Cart-inverted pendulum Outputs} shows that the system safety is achieved using three presented control schemes; however, one can see that the online adaptive control policy improves the system safety in comparison with the offline data-driven safe predictive control.

\begin{figure}[!ht]
     \centering
     \includegraphics[width=0.99\linewidth]{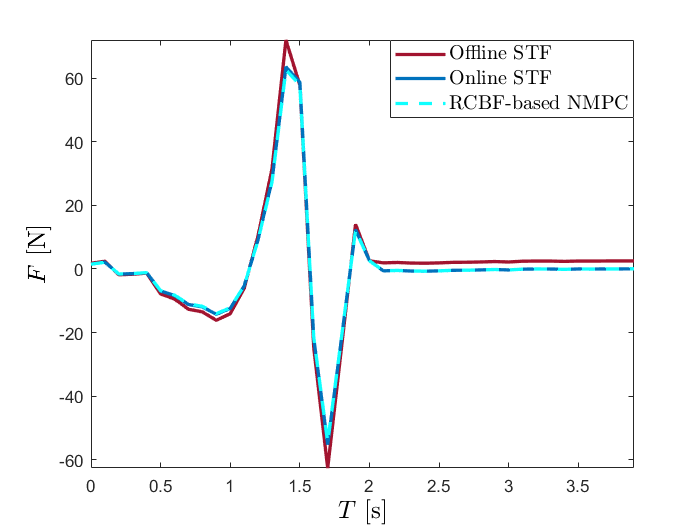}
    \caption{Control input for cart-inverted pendulum using RCBF-based NMPC and data-driven safe predictive control.}
     \label{Cart-inverted pendulum Control input}
 \end{figure}
 
 %\begin{figure}[ht]
 %    \centering
  %   \includegraphics[width=8.4cm, height=6.36cm]{Figs/StatesPendulum.eps}
%    \caption{States of cart-inverted pendulum using RCBF-based NMPC and data-driven safe predictive control.}
%    \label{Cart-inverted pendulum States}
% \end{figure}
 
 \begin{figure}[!ht]
     \centering
     \includegraphics[width=0.99\linewidth]{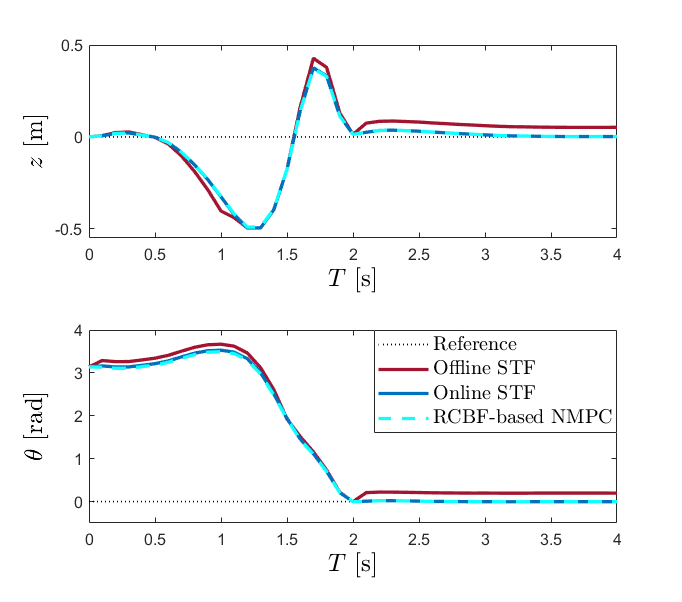}
    \caption{Outputs of cart-inverted pendulum using RCBF-based NMPC and data-driven safe predictive control.}
     \label{Cart-inverted pendulum Outputs}
 \end{figure}
 
  \begin{figure}[!ht]
     \centering
     \includegraphics[width=0.95\linewidth]{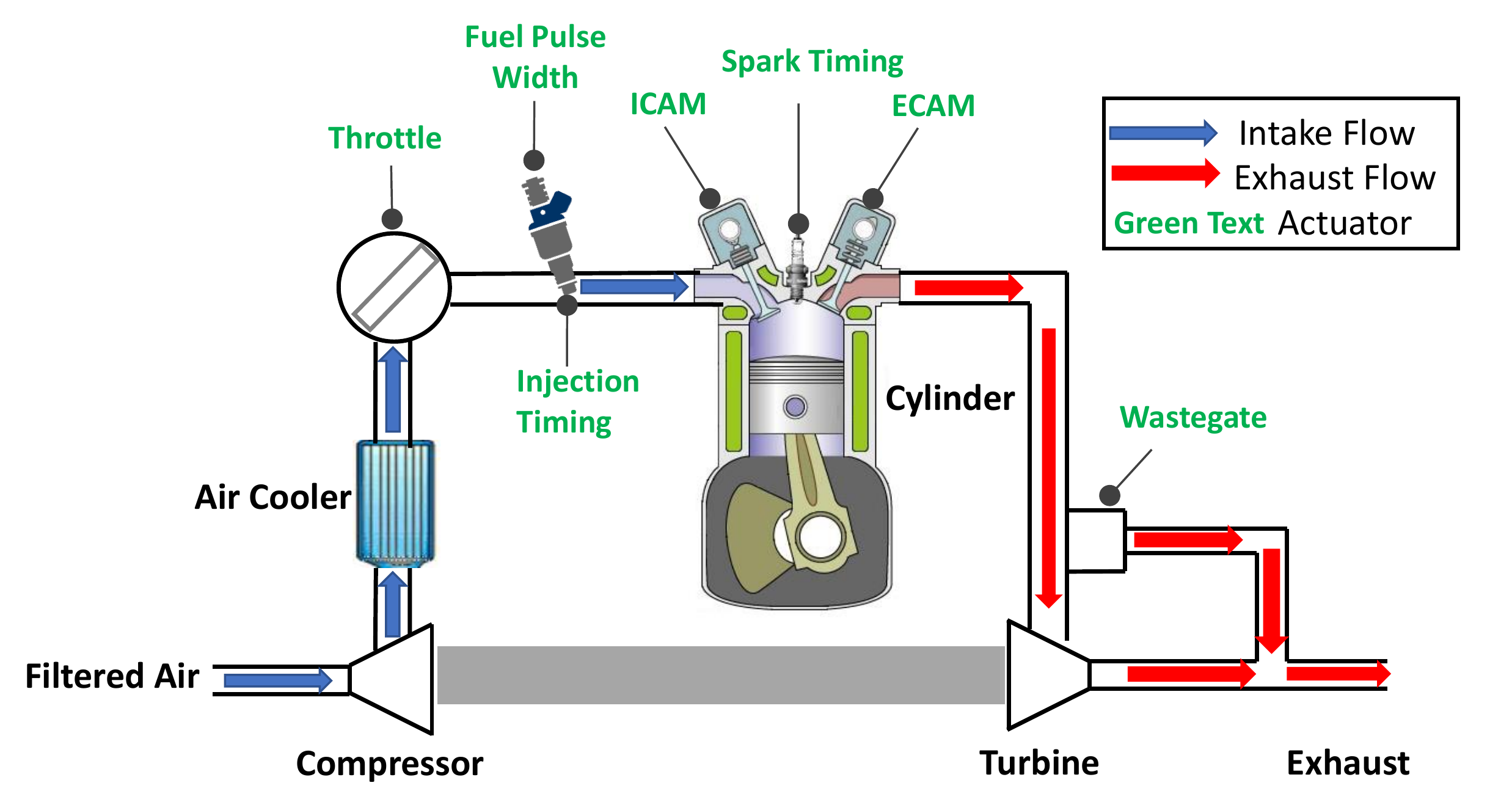}
    \caption{Turbocharged internal combustion engine.}
     \label{Turbocharged internal combustion engine}
 \end{figure}

\subsection{Turbocharged Internal Combustion Engine}
In this subsection, we apply our control framework on the turbocharged internal combustion engine as shown in Fig. \ref{Turbocharged internal combustion engine} \cite{chen2020online}. There are seven main actuators: throttle position, intake cam (ICAM) position, exhaust cam (ECAM) position, spark timing, fuel pulsewidth, fuel injection timing, and wastegate position. %The throttle and wastegate control the air charge to the engine intake manifold. Specifically, the throttle controls the airflow to the intake manifold, with larger throttle opening leading to more air into manifold. The wastegate controls the bypass exhaust flow to regulate the turbine speed, which subsequently regulates the boost pressure. The smaller opening of wastegate valve, the more exhaust flow driving the turbine, and the higher the boost pressure, and therefore, the more power the engine can generate. The ICAM and ECAM are used to control the gas exchange between the intake and exhaust with the cylinders and have significant and complicated impacts on the engine operations. Other than air charge, another important ingredient to generate engine torque is fuel. The fuel can be injected into the intake manifold or can be injected directly into the combustion chamber. The mass ratio of the air over fuel, so-called air-to-fuel ratio, is critical to the combustion process, which affects the engine efficiency and emission. The amount of fuel injection is controlled by the fuel pulsewidth and fuel pressure. The fuel pulsewidth is defined as the amount of time that the fuel injector stays open during an intake cycle. Higher fuel pulsewidth corresponds to longer fuel injection and consequently more fuel. The injection timing is the piston position at which the injector starts to inject fuel. It also has a great impact on the engine operation, but the relationship between power and emission is challenging to quantify. Another important factor in engine combustion is the spark timing, which refers to the piston at which the spark ignites the mixed air and fuel.
However, we use the first four actuators for STF-based nonlinear system identification as they play most significant roles on the gasoline engine controls. Moreover, the system outputs are considered as the fuel consumption rate and the vehicle speed. Apparently, the input–output relationship of this engine system is nonlinear and complex, which makes the system identification a challenging task. For this case, we aim to minimize the fuel consumption rate and track a reference for the vehicle speed.

According to Algorithm \ref{STF}, the online data-driven safe predictive control is applied to the turbocharged internal combustion engine based on data collected from a high-fidelity simulator provided by Ford, which yields 1) 96.12\% accuracy for 30000 training data with 20 clusters and local linear models in the offline system identification, and 2) 98.05\% accuracy for 100000 training data with 20 clusters and local linear models in the offline RCBF-based NMPC learning.

For this case, the RCBF-based NMPC policy \eqref{RCBF-based NMPC} works with the trained model obtained by the STF, and we adapt it using the feedback controller in \eqref{Online Adaptation} to minimize the performance loss caused by the the system identification error and the unknown disturbance for the real system. Considering four future state predictions $N=4$, Fig. \ref{Turbocharged internal combustion engine Control Inputs} shows the control input signals for the turbocharged internal combustion engine, i.e., the throttle position, the intake cam (ICAM) position, the exhaust cam (ECAM) position, and the spark timing. Like the previous example, one can see that the proposed online adaptive control policy corrects the offline data-driven safe predictive control and learns the RCBF-based NMPC policy better. Fig. \ref{Turbocharged internal combustion engine Outputs} shows the outputs of the turbocharged internal combustion engine, where the fuel consumption is minimized, and the vehicle speed tracks the desired reference while it satisfies the constraint. One can see that the offline data-driven safe predictive control causes a performance loss for the real system; however, the online data-driven safe predictive control minimizes the performance loss by removing the KKT deviations caused by the control learning error and state perturbations caused by the system identification error and unknown disturbance. Fig.~\ref{Number of clusters and local models} shows the distribution of the identification performance along the number of clusters and local models for the turbocharged internal combustion engine. As it is obvious from Fig. \ref{Number of clusters and local models}, there is no major change for the performance after 20 clusters; thus, we have considered this number for the engine vehicle identification. Moreover, Fig. \ref{Delay} illustrates the distribution of the identification performance along the input delay $d_u$ and output delay $d_y$ for the turbocharged internal combustion engine. One can see that $d_u=3,d_y=2$ makes the best identification performance for the engine vehicle, where these values are considered for the engine vehicle identification.

\begin{figure}[!ht]
     \centering
     \includegraphics[width=0.99\linewidth]{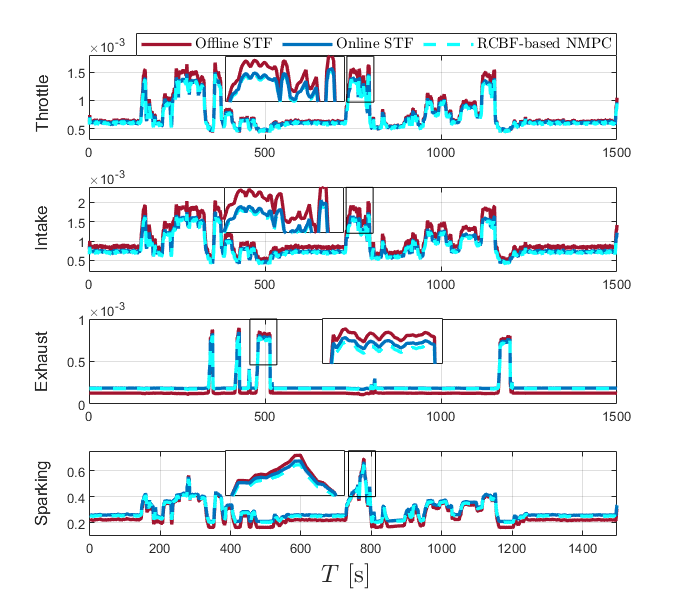}
    \caption{Control inputs for engine vehicle using RCBF-based NMPC and data-driven safe predictive control.}
     \label{Turbocharged internal combustion engine Control Inputs}
 \end{figure}

 \begin{figure}[!ht]
     \centering
     \includegraphics[width=0.99\linewidth]{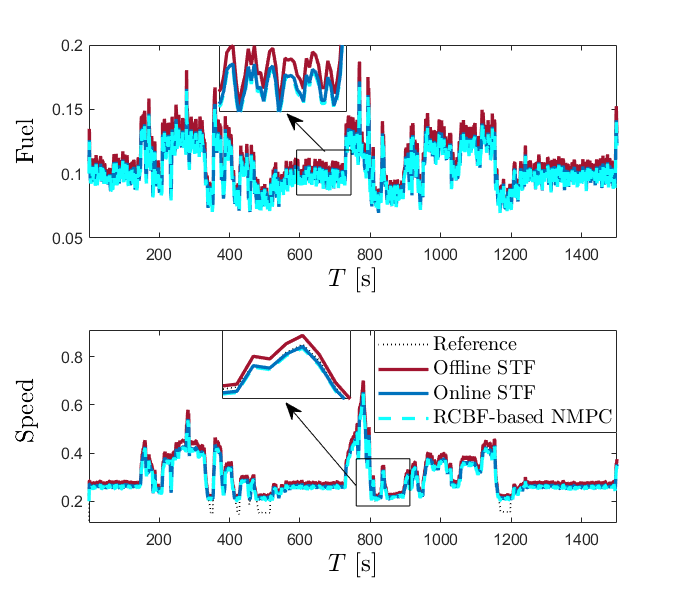}
    \caption{States for engine vehicle using RCBF-based NMPC and data-driven safe predictive control.}
     \label{Turbocharged internal combustion engine Outputs}
 \end{figure}
 
  \begin{figure}[!ht]
     \centering
     \includegraphics[width=0.99\linewidth]{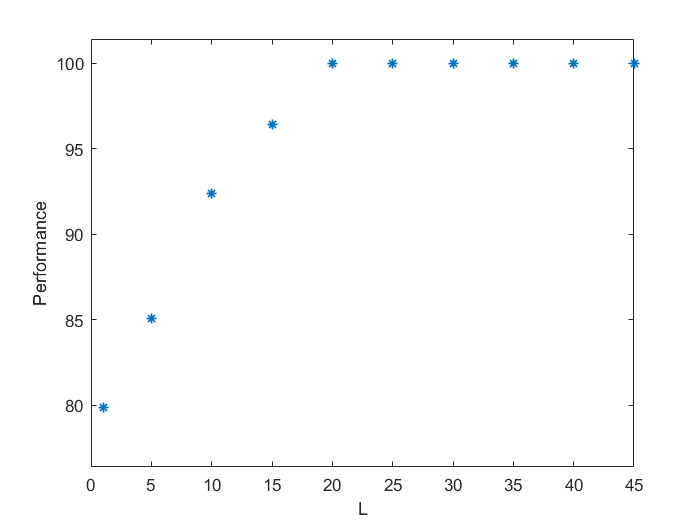}
    \caption{Number of clusters and local models for engine vehicle identification.}
     \label{Number of clusters and local models}
 \end{figure}
 
   \begin{figure}[!ht]
     \centering
     \includegraphics[width=0.99\linewidth]{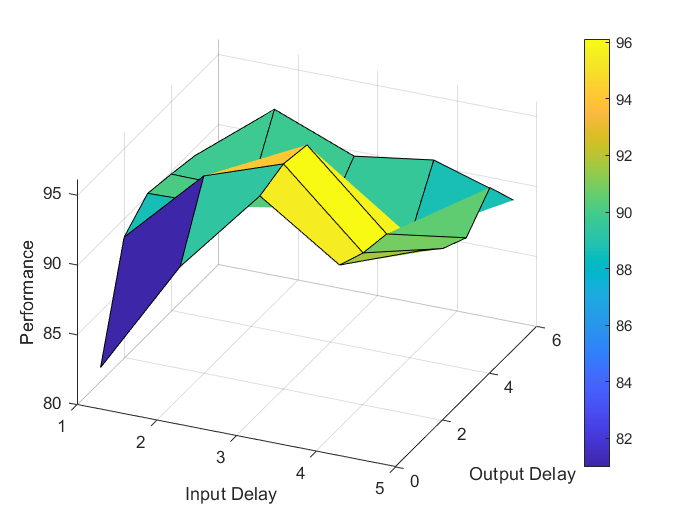}
    \caption{Input delay and output delay for engine vehicle identification.}
     \label{Delay}
 \end{figure} 
 
For the offline part, Table II presents the performance and computational cost of the discrete-time STF-based concurrent learning in comparison with the NNs and the GPR for the engine vehicle identification, and Table III demonstrates the same task for the RCBF-based NMPC policy learning. We evaluate the performance of each function approximator using the best fitting rate (BFR), which analyzes the goodness of fit between the validation data (i.e., measured data) and the simulated output of the trained model (or trained policy) based on the normalized root mean squared error (NRMSE). Thus, the performance represents $(1-\frac{\lVert y - \hat{y} \rVert}{\lVert y - mean(y) \rVert}) \times 100 \hspace{1 mm} \%$ and $(1-\frac{\lVert u_{mpc} - \tilde{u} \rVert}{\lVert u_{mpc} - mean(u_{mpc}) \rVert}) \times 100 \hspace{1 mm} \%$ for Tables II and III, respectively, and 100\% performance means that the simulated output from the trained model (or trained policy) is perfectly matched with the measured data. It is worth noting that the measured data for Tables II and III are the collected system output from the engine vehicle and the collected control input from the RCBF-based NMPC, respectively. As shown in Tables II and III, the STF-based concurrent learning effectively reduces the computational cost of the learning process while

\vspace{4 mm}
\begin{table}[!ht]
\centering
 \caption{Comparison of Performance and Computational Cost for System Identification}
\begin{tabular}{ |p{2.3cm}|p{2.3cm}|p{2.3cm}|  }
\hline\hline
Learning & Performance & Time (per loop) \\
\hline
NNs & $96.15\%$ & $32.8 \hspace{1 mm} ms$ \\\hline
GPR & $96.48\%$ & $99.3 \hspace{1 mm} ms$ \\\hline
STF & $96.12\%$ & $04.3 \hspace{1 mm} ms$ \\\hline
\hline
\end{tabular}
\end{table}
 
\vspace{-6 pt}
\begin{table}[!ht]
\centering
 \caption{Comparison of Performance and Computational Cost for NMPC Policy Learning}
\begin{tabular}{ |p{2.3cm}|p{2.3cm}|p{2.3cm}|  }
\hline
\hline
Fun. Forms & Performance & Time (per loop) \\
\hline
NNs & $98.12\%$ & $19.4 \hspace{1 mm} ms$ \\\hline
GPR & $98.57\%$ & $58.1 \hspace{1 mm} ms$ \\\hline
STF & $98.05\%$ & $03.5 \hspace{1 mm} ms$ \\\hline
\hline
\end{tabular}
\end{table}

\vspace{-6 pt}
\begin{table}[!ht]
\centering
 \caption{Comparison of Performance and Computational Cost for Different Controllers}
\begin{tabular}{ |p{3cm}|p{2cm}|p{2.3cm}|  }
\hline\hline
Control & Performance & Time (per loop) \\
\hline
NMPC (N=1) & $98.76\%$ & $05.7 \hspace{1 mm} ms$ \\\hline
NMPC (N=4) & $99.23\%$ & $25.2 \hspace{1 mm} ms$ \\\hline
NMPC (N=8) & $99.86\%$ & $68.1 \hspace{1 mm} ms$ \\\hline
Online STF (N=1) & $98.66\%$ & $66.2 \hspace{1 mm} \mu s$ \\\hline
Online STF (N=4) & $99.15\%$ & $78.9 \hspace{1 mm} \mu s$ \\\hline
Online STF (N=8) & $99.79\%$ & $90.8 \hspace{1 mm} \mu s$ \\\hline
\hline
\end{tabular}
\end{table}

\hspace{-3.5 mm}it shows high performance for system identification and NMPC learning compared to the NNs and the GPR. Moreover, for the online part, the online data-driven safe predictive control is compared with the RCBF-based NMPC for various future predictions $N$ in Table IV. This table demonstrates the performance of each controller using the NRMSE between the desired reference trajectory and the simulated system output (using the trained model). Thus, the performance represents $(1-\frac{\lVert r - \hat{y} \rVert}{\lVert r - mean(r) \rVert}) \times 100 \hspace{1 mm} \%$, and 100\% performance means that the simulated system output is perfectly matched with the desired reference trajectory. One can see that the online data-driven safe predictive control provides high performance to track the desired reference trajectory while it effectively reduces the computational cost of the NMPC.

\section{Conclusions}
\label{Sec5}
In this paper, we developed a unified online data-driven predictive control framework with robust safety guarantees, which includes a discrete-time STF-based concurrent learning for efficient nonlinear system identification with relaxed  PE conditions, a RCBF-based NMPC policy approximator to explicitly deal with  control learning error, system identification error, and  unknown disturbances, and an online adaptation law based on KKT sensitivity analysis and feedback control. The framework was applied to the cart-inverted pendulum as well as an automotive engine with promising results demonstrated. The main contribution of this paper is mainly on the proposed new framework with control law derivations and analysis. The main purpose of the simulation is to demonstrate the effectiveness of the proposed framework by showing that the developed online data-driven safe predictive control is able to achieve a reasonable performance with negligible online computation as compared to the NMPC. The considered cart-inverted pendulum is a classical system frequently used for nonlinear control benchmarks. Moreover, for the second example of simulation part, we have experimentally collected input/output (I/O) data from a turbocharged internal combustion engine and identified a nominal model for the system using the collected data. The proposed data-driven control is applied on the obtained identified model and demonstrates a reasonable performance as shown in the simulation results. We acknowledge a few limitations of our method as follows. Like the NNs and the GPR, the STF framework requires comprehensive data collection to ensure an adequate coverage of operating conditions. Moreover, in the control design, a bounded external disturbance is assumed, and we will extend this framework with more general unbounded stochastic disturbances in our future work. Another assumption we make is on the recursive feasibility of the optimization problem with control input and safety constraints to guarantee the closed-loop performance. Future work will include addressing the mentioned shortcomings by exploring a finite sample approach to reduce the required collected data for the STF and carrying out a formal discussion on the recursive feasibility of the optimization problem, e.g., with an iteration approach.
Finally, applications to real-world nonlinear and complex systems such as robots and autonomous vehicles will be performed.

\begin{appendices}
\section{(Transferring STF Model to State-Space Model)}
Considering \eqref{NARX Model}-\eqref{local models}, the input vector of the STF function approximator, i.e., ${U}_{stf}(k+1)$, can be written in the format of $[u(k);x(k)]$ as
\begin{equation}
  \begin{aligned}
    \label{local models state}
  &{f}_{i}(k+1) = {A}_{i}{U}_{stf}(k+1)+{b}_{i}+\omega_i(k+1) \\ 
  & \hspace{13.5 mm} = A_{i_2} x(k)+{A}_{i_1}u(k)+{b}_{i}+\omega_i(k+1),\\
  \end{aligned}
\end{equation}
where ${A}_{i_1}$ is the first element of ${A}_{i}$, and $A_{i_2}$ represents the rest of the elements, i.e.,
\begin{equation}
\small
  \begin{aligned}
    \label{parameters state}
    &{A}_{i} = [a_{i_1}, a_{i_2}, a_{i_3}, \ldots, a_{i_{d_u+d_y}}],\\
  &A_{i_1} = a_{i_1},{A}_{i_2} = [a_{i_2}, a_{i_3}, \ldots, a_{i_{d_u+d_y}}],\\
  &x(k) = [u(k-1); \ldots; u(k-d_{u}+1);y(k); \ldots; y(k-d_{y}+1)].
  \end{aligned}
\end{equation}
where $x(k)$ is considered as the states of the system, and $u(k)$ represents the control input.

Now, the nonlinear model \eqref{Composite Structure} is rewritten as
\begin{equation}
\small
  \begin{aligned}
    \label{Composite Structure state}
  y(k+1)=\sum\limits_{i=1}^{L} {\alpha }_{i}(k+1) (A_{i_2} x(k)+{A}_{i_1}u(k)+{b}_{i}+\omega_i(k+1)), 
  \end{aligned}
\end{equation}
where ${\alpha }_{i}(k+1)$ represents ${\alpha }_{i}([u(k);x(k)],{\psi }_{i})$.

Using \eqref{parameters state} and \eqref{Composite Structure state}, the state-space model \eqref{system} is obtained in the following form:
\begin{equation}
  \begin{aligned}
    \label{state-space}
    & x(k+1) = f(x(k),u(k))+w(k),\\
    & \hspace{13 mm} = A_{t_2}(k) x(k)+{A}_{t_1}(k) u(k)+{b}_{t}(k)+w(k),
  \end{aligned}
\end{equation}
where $A_{t_2}$, ${A}_{t_1}$, and ${b}_{t}$ are nonlinear matrices as
\begin{equation}
  \begin{aligned}
    \label{state-space matrices}
    & A_{t_2}(k)=\begin{bmatrix} 
    0_{n_u \times n_u(d_u -1)} & 0_{n_u \times n_y(d_y)}\\ 
    \tau_1 & 0_{n_u(d_u -2) \times n_y(d_y)}\\
    \rho_1(k) & \rho_2(k)\\
    0_{n_y(d_y -1) \times n_u(d_u-1)} & \tau_2
    \end{bmatrix},\\
    & A_{t_1}(k)=\begin{bmatrix} 
    I_{n_u}\\ 
    0_{n_u(d_u -2) \times n_u}\\
    \sum\limits_{i=1}^{L} {\alpha }_{i}([u(k);x(k)],{\psi }_{i}) {A}_{i_1},\\
    0_{n_y(d_y -1) \times n_u}
    \end{bmatrix},\\
    & b_{t}(k)= \begin{bmatrix} 
    0_{n_u \times 1}\\ 
    0_{n_u(d_u -2) \times 1}\\
    \sum\limits_{i=1}^{L} {\alpha }_{i}([u(k);x(k)],{\psi }_{i}) b_i,\\
    0_{n_y(d_y -1) \times 1}
    \end{bmatrix},\\
    & w(k)= \begin{bmatrix} 
    0_{n_u \times 1}\\ 
    0_{n_u(d_u -2) \times 1}\\
    \sum\limits_{i=1}^{L} {\alpha }_{i}([u(k);x(k)],{\psi }_{i}) \omega_i(k+1),\\
    0_{n_y(d_y -1) \times 1}
    \end{bmatrix},\\
  \end{aligned}
\end{equation}
where
\begin{equation}
  \begin{aligned}
    \label{state-space terms}
    & \tau_1 = [I_{n_u(d_u -2)} \hspace{5 mm} 0_{n_u(d_u -2) \times n_u}],\\
    & \tau_2 = [I_{n_y(d_y -1)} \hspace{5 mm} 0_{n_y(d_y -1) \times n_y}],\\
    & \rho_1(k) = \sum\limits_{i=1}^{L} {\alpha }_{i}([u(k);x(k)],{\psi }_{i}) [a_{i_2}, \ldots, a_{i_{d_u}}],\\
    & \rho_2(k) = \sum\limits_{i=1}^{L} {\alpha }_{i}([u(k);x(k)],{\psi }_{i}) [a_{i_{d_u+1}}, \ldots, a_{i_{d_u+d_y}}].
  \end{aligned}
\end{equation}
\end{appendices}

\bibliographystyle{ieeetr}
\bibliography{References.bib}
 
\end{document}